%% file: arxiv_main.tex
\documentclass[11pt]{article}
\usepackage{amsmath,amsthm,amsfonts}
\usepackage{amsmath,amsfonts,amssymb}
\usepackage{geometry}
\usepackage{mathabx}
\usepackage{appendix}
\usepackage{snaptodo}
\usepackage[plain,noend,linesnumbered]{algorithm2e}
\SetKwProg{Proc}{Procedure}{}{}
\SetKwProg{Fn}{Function}{}{}
\usepackage{bbm}
\usepackage{authblk}
\usepackage{hyperref}
\usepackage[left]{lineno}
\usepackage{bm}
\usepackage{comment}
\usepackage{natbib}
\usepackage{caption}
\renewcommand{\sqrt}[1]{\left\{ #1 \right\}^{1/2}}

\newtheorem{theorem}{Theorem}[section]
\newtheorem{lemma}[theorem]{Lemma}

\newtheorem{definition}[theorem]{Definition}

\newtheorem{proposition}[theorem]{Proposition}

\newtheorem*{assumptions*}{Assumptions}

\input{newc}

\usepackage{snaptodo}
\title{An Optimal False Discovery Rate Controlling Procedure for Changepoint Detection}

\author[1]{Louis Davis}
\author[1]{Guenther Walther}
\affil[1]{Department of Statistics, Stanford University, 390 Jane Stanford Way, Stanford, California, U.S.A., 94305\\ \texttt{\{ldavis2, gwalther\}@stanford.edu}}

\date{\today}

\begin{document}

\maketitle

\begin{abstract}

We consider the problem of detecting and localizing a growing number of changepoints in a sequence of independent observations. We propose a new method, Lean Bonferroni Detection - False Discovery Rate (LBD-FDR), which produces a set of localized regions in the data sequence where, in expectation, a high proportion of them contain a changepoint. LBD-FDR guarantees control of the false discovery rate in a wide range of distributional settings, including the challenging case of independent non-parametric and heavy-tailed data. For independent Gaussian sequences, we derive conditions on changepoint arrangements where the method consistently detects all changepoints with a large enough signal while simultaneously being unaffected by those that are undetectable, and we show that LBD-FDR obtains the optimal detection constant in certain regimes. Moreover, we derive the settings where our method is more powerful than the minimax optimal Type I error controlling method. We finally develop a computationally feasible algorithm for the LBD-FDR and compare it in simulation to five existing procedures.

\end{abstract}

\section{Introduction}\label{sec:intro}

Changepoint detection is an area of statistics with a rich history going back to the classical works of \cite{wald1945sequential} and \cite{Page1954}. It is currently experiencing a renaissance due to the proliferation of such problems in modern data analysis, see e.g. the review articles of  \cite{Aue2013,Jialing2025,Niu2016}. Many detection and localization methods exist in the literature, with most being designed to provide estimates for change-points and then construct confidence intervals around the estimated change-points to quantify uncertainty. In contrast, building on \cite{FryzlewiczNSP2024} and \cite{jang2024fastoptimalchangepointdetection} we propose a new methodology for determining tight intervals for the locations of changepoints and we provide a finite sample guarantee for the coverage of these intervals. 

In particular, consider the piecewise constant signal plus noise model
\[X_k = \mu_k+\varepsilon_k, \quad k=1,\dots,n\]
where the $\varepsilon_k$ are independent centered noise and $\mu_k$ is a piecewise constant vector with an unknown number $N_n$ of changepoints at locations $0=:\tau_0<\tau_1<\dots<\tau_{N_n}<\tau_{N_n+1}:=n$ such that $\mu_{\tau_{k}}\neq \mu_{\tau_{k+1}}$. Our method provides a collection of $\hat N_n$ intervals such that for any given $\alpha \in (0,1)$ we expect the proportion of intervals that do not contain a changepoint to be at most $\alpha$. 

We now review the literature on uncertainty quantification in multiple changepoint detection problems. SMUCE \citep{SMUCE2014}, and its more recent variants such as H-SMUCE \citep{HSMUCE2017} and MQS \cite{Vanegas03072022}, estimate the number $N$ of changepoints as the minimum among candidate models $\hat \mu$ for which the empirical residuals do not violate a global hypothesis at level $\alpha.$ Especially for small $\alpha$, the theoretical coverage properties of the estimate, $\hat \mu$, rely on correctly estimating the true number of changepoints, which may not be feasible in scenarios where one true changepoint has a weak signal. 
NSP \cite{FryzlewiczNSP2024} and RNSP \cite{FryzlewiczRNSP2024} move the inferential target from estimating the number of changepoints to regions in which they lie. They do so via a recursive algorithm that segments $\{1,\dots,n\}$. When $n$ is large, it searches a representative sample of a user selected number $M$ of intervals for deviations from a linear model. However, this approach will not result in optimal detection power, as explained in \citep{kovacs2022,WaltherPerry2021}. In contrast, LBD \citep{jang2024fastoptimalchangepointdetection} constructs a deterministic set of intervals, tests the existence of a changepoint in each interval, and then combines the local tests on each interval globally using a weighted Bonferroni correction. By selecting an appropriate set of intervals that is sparse enough so that the Bonferroni correction is not overly conservative, while also being sufficiently flexible to the unknown changepoint arrangement, they can guarantee optimal changepoint detection.

While the method of \cite{jang2024fastoptimalchangepointdetection} is minimax optimal among all Type I error controlling algorithms, when the signal-to-noise ratio is modest and the number of changepoints is large, controlling Type I error may be too conservative, in direct analogy to the general hypothesis testing setting \citep[e.g.][]{Storey2003,Sarkar2007}. Instead, controlling the false discovery rate \citep[e.g.][]{BH1995} may allow for smaller and more frequent signals to be detected at the same sample size. This intuition was the impetus behind the development of the false discovery rate controlling methods FDRSeg \citep{Li2016FDRSeg} and MUSCLE \cite{liu2025multiscalequantileregressionlocal}. Compared to their Type I error controlling counterparts SMUCE and MQS, FDRSeg and MUSCLE estimate the signal $\mu$ by ensuring the empirical residuals do not violate local multiscale side constraints compared to the global constraints of SMUCE and MQS. In addition to stronger theoretical assumptions than we will make below, the latter two methods do not provide tight regions around the estimated change-points where the true changepoint may lie with reasonable ``confidence''. Note that this is not true statistical confidence since the methods are obtained by inverting a FDR-controlling (rather than Type I error controlling) procedure. 

To these ends, we introduce {\sl Lean Bonferroni Detection - False Discovery Rate (LBD-FDR)}. This method evaluates local tests on the same set of intervals as LBD and combines them not with a Bonferroni correction, but rather with a recently introduced graphical Benjamini-Hochberg procedure known as IndBH \citep[][]{nguyen2025controllingfalsediscoveryrate}. We demonstrate that our method has provable finite-sample false discovery rate control, at given level $\alpha$, for ``tight'' regions in which a change-point may lie. Our finite sample validity result holds in many distributional settings ensuring that our method is robust to slight model misspecifications. More generally, our method is valid in non-parametric settings where we only assume exchangeability of the data in an interval containing no changepoint. It can also be specialized to exponential family settings.

We prove a consistency theorem for ``detectable'' changepoints in the canonical mean-shift Gaussian model. In contrast to many existing methods including SMUCE, H-SMUCE, FDRSeg, MQS and MUSCLE, we demonstrate that these theoretical guarantees of our method hold in the presence of undetectable changepoints \citep[see Section 2.3][for a further discussion]{Verzelen2023}. Following this theorem, we compare the energy that is sufficient for the detection and localization of a changepoint to that of \cite{jang2024fastoptimalchangepointdetection} and find that LBD-FDR can detect changepoints $\tau_k$ having an energy that falls below the minimax boundary for Type I error controlling procedures if $(\tau_k-\tau_{k-1})\wedge (\tau_{k+1}-\tau_k)<N_n$. This regime, where there are many changepoints or they are closely spaced, is where many state of the art methods have difficulty \citep{Fryzlewicz2020} and so this threshold alone demonstrates the merits of our method.

Following our consistency results, we discuss impossibility theorems that give conditions under which no false discovery rate controlling procedure can localize all $N_n$ changepoints with non-trivial probability. In particular, we derive the signal strength necessary for detection. Under weak assumptions on the changepoint arrangements, roughly that the spacing and total number of change-points are sub-polynomial, we demonstrate that LBD-FDR attains this detection threshold and therefore is optimal in this regime. 

Finally, we improve the practical implementation of the IndBH procedure, which LBD-FDR employs for a particular interval system. In general, IndBH relies on solving a NP-hard problem which in many multiscale scenarios is computationally infeasible. However, using the properties of interval overlap graphs, we can exploit dynamic programming to massively decrease the time and memory complexity of the procedure. In the worst case, the time complexity is $O(n^2)$ up to a polylogarithmic factor. However, when the data come from an exponential family and the signal-to-noise ratio is not too large, this reduces to nearly linear time complexity. In a simulation study we compare our method with five existing methods: LBD, SMUCE, FDRSeg, MUSCLE, and MQS. For a fixed sample size, we vary the number of changepoints or the distribution of the noise, including cases where the noise has infinite mean and/or variance, which demonstrates that our procedure is powerful, valid and can produce tight localization regions for a true changepoint.

\section{The Statistical Setting} \label{StatSetting}

We consider the problem of detecting and localizing the changepoints in the distribution of the $n$ independent observations $\{X_t\}_{t=1}^n$, where $\tau$ is called a changepoint if $X_\tau \stackrel{d}{\neq}X_{\tau+1}$. We let $\mathcal T$ be the set of all $N_n$ change-points and index the potentially growing number of them so that $0=:\tau_0<\tau_1<\tau_2<\dots<\tau_{N_n}<\tau_{N_n+1}=:n$. A special case of this setting is the canonical Gaussian mean shift model\
\begin{equation*}
    X_k=\mu_k+\varepsilon_k,\quad k=1,\dots,n,
\end{equation*}
where $\varepsilon_k\stackrel{\text{iid}}{\sim}N(0,1)$ and $\mu=(\mu_1,\ldots,\mu_n)$ is a piece-wise constant deterministic signal with unknown indices of change $\tau \in \mathcal T$. But we will also allow more general distributional settings below.

Detecting multiple changepoints can be phrased as a multiple testing problem of the local hypotheses $H_{0,k}: X_k\stackrel{d}{=}X_{k+1}$ vs. $H_1: X_k\stackrel{d}{\neq}X_{k+1}$, for $k\in \{1,\dots,n-1\}$. Following \cite{jang2024fastoptimalchangepointdetection} we attack this  problem by considering a collection of intervals $\mathcal I$ and analyzing the multiple testing problem $H_{0,j}': \{X_i\}_{j\in I_j}$ are i.i.d. (or exchangeable), for $I_j \in \mathcal I$. The motivation of this approach is that rejecting $H_{0,j}'$ implies the existence of a changepoint in $I_j$. Therefore the statistical guarantees for this multiple testing problem apply to both the existence and the localization of the changepoint.  In order to maximize the power for detecting a changepoint $\tau_k$, one needs to test a stretch of data in $(\tau_{k-1} ,\tau_{k+1}]$ that is as long as possible. Thus $\mathcal I$ needs to be sufficiently rich. On the other hand, if
$\mathcal I$ is too large (such as the collection of all intervals in $\{1,\ldots,n\}$), then the  penalty for the resulting large multiple testing problem may result in suboptimal inference. Therefore the choice of the collection $\mathcal I$ is crucial for achieving optimal detection.

The quantity determining how easy it is to detect the changepoint $\tau_k$ (heuristically this is the strength of the signal) is called the {\sl energy of $\tau_k$} \citep[e.g.][]{Verzelen2023}:
\begin{align}\label{eq:energydefCP}
    \mathcal E_k=|\mu_{\tau_k +1}-\mu_{\tau_k}|\sqrt{\frac{(\tau_k-\tau_{k-1})(\tau_{k+1}-\tau_k)}{\tau_{k+1}-\tau_{k-1}}}.
\end{align}

A related quantity is the minimum energy over all changepoints,
\begin{align}\label{eq:minenergy}
    \mathcal E_{\min}=\left(\min_{1\leq k \leq N_n}|\mu_{\tau_k+1}-\mu_{\tau_k}|\right)\left(\min_{0\leq k\leq N_n}(\tau_{k+1}-\tau_k)^{1/2}\right)
\end{align}
investigated by \cite{Hao2013-rg,WBS2014,SMUCE2014,Frylcewizc2018,kovacs2022,Vanegas03072022,liu2025multiscalequantileregressionlocal}. If the changepoints are homogeneous, then $\min_k \mathcal E_k\leq \mathcal E_{\min}$. But if the changepoints are heterogeneous, then $\mathcal E_{\min}$ may be considerably smaller than the energy of each individual changepoint. In the canonical Gaussian setting, the aforementioned consistency theorems typically require $\mathcal E_{\min}\geq C\sqrt{\log(n)}$, for some constant $C>0$. This condition has two important shortcomings \citep[c.f.][]{Verzelen2023,jang2024fastoptimalchangepointdetection}. The first is that  if at least one changepoint that does not satisfy the energy condition, then no guarantees are available even if some other changepoints are easily detectable. The second is that the important quantity for determining the difficulty of the detection problem is the constant $C$ rather than the rate $\sqrt{\log(n)}$, see the discussion in \cite{jang2024fastoptimalchangepointdetection}.  As a simple example, if the detection threshold is $\sqrt{2\log(n)}$, then erroneously declaring that the detection threshold is $\sqrt{3\log(n)}=\sqrt{2\log\left(n^{3/2}\right)}$ is equivalent to observing a sample size equal to $n^{3/2}$ rather than $n$. Thus, for detection and localization problems, optimality is concerned with the constant, and not the rate, of the energy for each changepoint \citep[e.g.][]{Castro2005,walther2024calibrating}. Consequently, our necessary and sufficient conditions for changepoint detection will be concerned with achieving the optimal constant $C$ for the energy $\mathcal E_k $. 

\section{Lean Bonferroni Detection and the False Discovery Rate}\label{sec:LBDandFDR}

In this section, we present our method by first reviewing Bonferroni triplets, which we use to define stretches of the data on which we test the i.i.d. hypothesis locally. We then review the IndBH procedure, which we use to combine the local tests, and present our method as applying the IndBH procedure to the Bonferroni triplets. We conclude this section with theorems that establish the finite sample validity for three versions of our procedure.

Bonferroni triplets \citep{jang2024fastoptimalchangepointdetection} are a natural extension of Bonferroni intervals introduced by \cite{walther2010, WaltherPerry2021}. For any integer $\ell \in \{0,\dots,\ell_{\max}:=\lfloor \log_2(n/4)\rfloor-1\}$, intervals of lengths $L_I\in [2^\ell,2^{\ell+1})$ are approximated by the collection of intervals
\begin{equation}
    \mathcal J_\ell=\left\{(j,k]\ \mid \ j,k\in \{id_\ell,i=0,\dots\},2^\ell \leq k-j<2^{\ell+1} \right\}
\end{equation}
where the grid spacing is $d_\ell=\left\lceil 2^\ell\left\{2\log\left(\frac{en}{2^\ell} \right) \right\}^{-1/2}\right\rceil$. The collection $\bigcup_\ell \mathcal J_\ell$ is the set of all Bonferroni intervals. Bonferroni triplets are constructed by adjoining an additional interval to the left or right of a Bonferroni interval such that the length of the adjoined interval is at least as large as the Bonferroni interval. To be precise, define $\mathcal L_n:=\{\text{lengths of all Bonferroni intervals}\}$, then the collection of Bonferroni triplets is $\bigcup_{\ell=0}^{\ell_{\max}}\mathcal K_\ell$ where
\begin{multline}\label{eq:BonfTrip}
\mathcal K_\ell=\bigg\{(t_1,t_2,t_3) \ \mid \ (t_1,t_2] \in \mathcal J_\ell \text{ and } t_3-t_2\in \mathcal L_n, t_3-t_2\geq t_2-t_1, \ t_3\leq n\\
\text{ or } (t_2,t_3] \in \mathcal J_\ell \text{ and } t_2-t_1\in \mathcal L_n, t_2-t_1> t_3-t_2, \ t_1\geq 0\bigg\}.
\end{multline}

The idea behind the construction of Bonferroni triplets is to approximate any triplet of changepoints $(\tau_{k-1},\tau_k,\tau_{k+1})$ by a Bonferroni triplet $(s,m,e)$ well enough so that the stretches of data $(X_{s+1},\ldots,X_m)$, $(X_{m+1},\ldots,X_e)$ allow for optimal testing of a changepoint at $\tau_k$, while at the same time the collection of Bonferroni triplets is sparse enough to keep the size of the multiple testing problem sufficiently small for optimal simultaneous inference.

As shown in \cite{jang2024fastoptimalchangepointdetection}, $\left| \bigcup_{\ell=0}^{\ell_{\max}}\mathcal K_\ell \right|=O\left(n\log^{5/2}(n) \right)$ so that the computational complexity for most test-statistics is linear up to polylogarithmic factors. For each Bonferroni triplet $t=(s,m,e)$,  testing ``there is no change-point in $(s,e]$'' against ``there exists a change-point at $m$'' in the parametric setting is most naturally done using  the generalized likelihood ratio test statistic on $(s,e]$. In particular, if $X_i \sim f_{\theta_i}$,  where $f_\theta$ is the distribution for some parametric family, then we will use the test-statistic

\begin{align}\label{eq:TtX}
    T_t( X)&=\sqrt{2\log LR_t( X)}\\ \label{eq:LRtX}
    \text{LR}_t( X)&=\frac{\left\{ \sup_{\theta \in \Theta}\prod_{i=s+1}^mf_\theta(X_i)\right\}\left\{\sup_{\theta \in \Theta}\prod_{i=m+1}^ef_\theta(X_i)\right\}}{\sup_{\theta \in \Theta}\prod_{i=s+1}^ef_\theta(X_i)}.
\end{align}
\cite{Walther2022TailBounds} demonstrates this transformation of the generalized likelihood ratio statistic yields distribution-free tail bounds in the case when $f_\theta$ belongs to an exponential family. For a more comprehensive treatment of other distributional settings, such as nonparametric or heavy tails using the Wilcoxon rank-sum statistic, see Section 2.3 of \cite{jang2024fastoptimalchangepointdetection}. 

We now outline the two statistics used in our simulation study (see Section \ref{sec:simulations}). In the known variance Gaussian case (WLOG $\sigma=1$), the generalized likelihood ratio statistic is
\begin{align}\label{eq:CUSUMGaussian}
Z_t =  \sqrt{\frac{(m-s)(e-m)}{e-s}}\big(\bar X_{s+1:m}-\bar X_{m+1:e}\big),\quad 
T_t( X)=\sqrt{2\log \text{LR}_t( X)}\ =\ |Z_t|,
\end{align}
which is equivalent (up to scale) to the two-sample $t$-statistic on $(s,e]$ with a split at $m$ when the variance is known. 

When the data have heavy tails, we use the two-sample Kolmogorov-Smirnov statistic \citep[e.g.][and references cited therein]{Pratt1981}. Informally, we are testing that the data in $(s,m]$ and $(m,e]$ are drawn from the same distribution $F$. Formally, we perform this test using the standard teststatistic
\begin{align*}
    D_{(s,m,e]}=\sup_{t\in \R}\left|\frac{1}{m-s}\sum_{i=s+1}^m\mathbbm{1}\{X_i\leq t\}-\frac{1}{e-m}\sum_{i=m+1}^e\mathbbm{1}\{X_i\leq t\}\right|
\end{align*}
which is the largest departure between the empirical cumulative distribution functions of the samples in the left and right halves of the triplet $(s,m,e]$. The behavior of this test statistic has been well studied, and following standard protocol we invert the test statistic to create a $p$-value exactly for small data sizes, i.e. $(m-s)(e-m)\leq10^4$, and we use asymptotic formulas otherwise, as is done in base R. 

\cite{jang2024fastoptimalchangepointdetection} detect and localize changepoints by applying local tests on Bonferroni triplets and then combining them using a weighted Bonferroni correction. They obtain finite sample Type I error rate guarantees, in addition to minimax detection power with optimal constants: $2^{1/2}$ for a bounded number of changepoints and $2^{3/2}$ for many changepoints with roughly equal spacing. As discussed in the introduction, controlling the Type I error may be too conservative in some scenarios, so that it becomes attractive to instead control the false discovery rate. However, to do so powerfully is a non-trivial task due to the dependence structure of the local hypotheses which are clearly two-sided. Even in the case where the errors are independent and identically distributed, $p$-values computed on two overlapping intervals, $(s_i,e_i]$ and $(s_j,e_j]$, are dependent and so may not be positive regression dependent on a subset \citep{Fithian2022CondBH}. It is not known if the Benjamini-Hochberg procedure \citep{BH1995} will control the false discovery rate in such a case. The Benjamini-Yekutieli procedure \citep[e.g.][]{BYekutieli2001} will, but it may be overly conservative and it has recently been shown to be inadmissible \citep{xu2025bringingclosurefdrcontrol}.

Fortunately, in our case where the data form an independent sequence, the dependency structure of the test statistics can be encoded graphically. In particular, we will use the independent set Benjamini-Hochberg procedure (IndBH) from \cite{nguyen2025controllingfalsediscoveryrate} to combine our local tests. Informally, IndBH is a method that controls the false discovery rate of dependent tests by exploiting the graphical dependence structure to find independent tests. Formally, suppose that $p \in [0,1]^m$ is an $m$-vector of $p$-values, and let $\D$ be an undirected graph with nodes $\{1,2,\dots,m\}$. Moreover, for each node of $\D$, let $N_{\D,i}\subset [m]$ be the set of nodes that have a shared edge with $i$. 

\begin{definition}
     We call $\D$ a dependency graph for the distribution $P$ of the $m$-vector $p$ if $p_i$ is independent of $\{p_j\}_{j \in N_{\D,i}^c}$ for every $i$.
\end{definition}
Following the convention in \cite{nguyen2025controllingfalsediscoveryrate} we state that $\D$ contains all self-edges since $p_i$ is clearly dependent on itself. If two nodes in a dependency graph do not share an edge, they are marginally independent.  

The IndBH procedure is a specific case of a graphical false discovery rate controlling method. It is defined by applying the Benjamini-Hochberg procedure on a masked $p$-value vector for every independent set in $\D$ and then taking the union of the resulting rejections. To be clear, an independent set of a graph is a set of  vertices such that two distinct vertices do not share an edge between them. If $I$ is an independent set, then the $\{p_i\}_{i \in I}$ are mutually independent. We denote the collection of all independent sets as Ind($\D$). In the following, let $\mathcal R_\alpha^\psi(p)$ be the rejection set for the valid testing procedure $\psi$ at level $\alpha$. Furthermore, let $1^Ap$ be the masked vector of the $p$-values, defined as
\[(1^Ap)_i=\begin{cases}
    1, \quad i \in A\\
    p_i, \quad i \in A^c.
\end{cases}
\]
Then, the IndBH rejection set is
\[\mathcal R_\alpha^{\text{IndBH}_\D}(p)= \bigcup_{I \in \text{Ind}(\D)}\mathcal R_\alpha^{\text{BH}}\left(1^{I^c}p\right).\]
Depending on the sparsity of the graph $\D$, the IndBH procedure interpolates between the Benjamini-Hochberg and Bonferroni procedures. Specifically, when the graph $\D$ has no off-diagonal edges, so that all of the $p$-values are independent, then some basic considerations show that the Benjamini-Hochberg procedure is recovered, whereas when $\D$ is complete, all $p$-values are dependent and the Bonferroni procedure is recovered. Analogously to the classical multiple testing literature, we expect the power of this procedure to be greater than applying a Bonferroni correction in the case where there are many non-null hypotheses, in which case Type I error controlling procedures have unsatisfactory power \citep[e.g.][]{BH1995,BYekutieli2001}.

Our changepoint detection and localization method {\sl LBD-FDR} is to first compute local test-statistics and their $p$-values, where `local' means that we use only a stretch of the data specified by a single Bonferroni triplet. We then encode the dependency structure of the $p$-values graphically. Specifically, node $i$, corresponding to $p_i$ (or equivalently triplet $i$) is connected to node $j$, equivalently $p_j$ or triplet $j$, if, and only if, the data stretches pertaining to triplets $i$ and triplet $j$ have non-empty intersection, i.e.\ $(s_i,e_i]\cap (s_j,e_j]\neq \emptyset $. Using this interval overlap dependency graph $\D$, we can then combine the local tests using the IndBH procedure. 

By construction, this method is valid in that it controls the false discovery rate over the entire family of triplets under few assumptions. Moreover, as in the case of the Lean Bonferroni Detection, we can make validity claims on two pruned rejection sets: First, the largest subset of rejected $p_i$ such that the corresponding intervals $(s_i,e_i]$ are disjoint, and second, the subset of  rejected $p_i$ such that the corresponding intervals $(s_i,e_i]$ are minimal with respect to inclusion, i.e. $(s_i,e_i]$ will not strictly contain $(s_k,e_k])$ for a rejected $p_k$.   We point out that in general, deterministic pruning of rejected hypotheses is not guaranteed to control the false discovery rate \citep[e.g.][]{Katsevich02012023,nair2026diversifyingconformalselections}. However, here we can exploit the richness of Bonferroni triplets and the certificate set structure of IndBH to prove validity. The proof relies on the fact that a certificate set of independent hypotheses corresponds to disjoint triplets, so that we can obtain a non-trivial lower bound on the number of minimal (or disjoint) intervals that are rejected.

\begin{theorem}\label{thm:prunedFDRcontrol}
     Suppose that $\{X_t\}_{t=1}^n$ is an independent sequence of observations, and that valid $p$-values are used for testing the existence of a changepoint in each of the Bonferroni triplets defined in \eqref{eq:BonfTrip}. That is, the local test corresponding to a Bonferroni triplet $(s,m,e)$ only uses the observations $(X_{s+1},\ldots,X_e)$, and the resulting $p$-value is marginally super-uniform under the null hypothesis that no changepoint exists in $(s,e)$. 

     For $\square \in \{\text{all},\min,\text{disj}\}$ define $\mathcal R_\alpha(\square)$ to be the rejection set obtained by LBD-FDR corresponding to all of the rejected intervals, the subset of rejected intervals that are minimal with respect to inclusion, and the largest subset of disjoint rejected intervals, respectively.

     Finally, define $\mathcal H_0$ to be the set of intervals $(s,e]$ that are induced by some Bonferroni triplet $(s,m,e)$ and that do not contain a changepoint in $(s,e)$. Then for any $n\geq 4$ and every $\alpha \in (0,1)$, the LBD-FDR procedure controls the false discovery rate for the collections of minimal, disjoint or all intervals at level $\alpha$. In particular,
     \begin{equation}\label{eq:FDRControl}
         \text{FDR}^{\square}=E\left[\frac{1}{|\mathcal R_\alpha(\square)|\vee 1}\sum_{J\in \mathcal H_0}\mathbbm{1}\left\{J\in \mathcal R_\alpha(\square)\right\}\right]\leq \frac{\alpha |\mathcal H_0|}{m_n}\leq \alpha
     \end{equation}
     where $m_n$ is the total number of Bonferroni triplets, and hence the number of hypotheses.
\end{theorem}

\begin{proof}
    If $\square=\text{all}$, equation \eqref{eq:FDRControl}  immediately holds under the stated assumptions by applying Proposition 3 of \cite{nguyen2025controllingfalsediscoveryrate}. In more detail, for our set of $m_n$ $p$-values, $p_i$ is connected to node $p_j$ if, and only if, $(s_i,e_i]\cap (s_j,e_j]\neq \emptyset $, and due to the construction of our local test statistics, a set of p-values $\{p_i\}$ is mutually independent if the corresponding $(s_i,e_i]$ are disjoint. The other two cases can be found in Appendix \ref{sec:Pruningproof}.
\end{proof}

 The set of minimal  intervals is of interest for localizing changepoints, as these are the shortest intervals that can detect a signal inside of them, and so they form a family of short ``confidence'' regions. The set of disjoint intervals can be used for inference on the total number of changepoints in the data. For example, if there are 100 rejected disjoint intervals, then in expectation, at most 100$\alpha$ of these intervals do not contain a changepoint. Since these intervals are disjoint, we therefore expect there to be at least $100(1-\alpha)$ changepoints in the data. Of course, this is not a true lower confidence bound, but as will be seen in Section \ref{sec:PowerandOptimality}, losing this confidence interpretation may be worthwhile, as in exchange, the sufficient local detection energy sometimes falls below that required by Type I error controlling procedures.  

In addition to the finite sample validity for pruned rejection sets, many graph theoretic calculations (such as computing the size of maximal independent sets or clique sizes) are facilitated by the interval overlap structure of the dependency graph. These improvements are made possible by dynamic programming, which can be applied after sorting the start and end points of every interval in the family. We exploit this structure to improve the computational efficiency of IndBH in Section \ref{sec:Implementation} and Appendix \ref{subsec:compdetails}.

\section{Sufficient Detection Energy, Impossibility and Segment Detection}\label{sec:PowerandOptimality}

Throughout this section, we investigate the power of the LBD-FDR in a Gaussian setting. We first derive the sufficient energy for detection and then compare it to that of the LBD-FWER. Following this, we discuss when our upper bound is sharp, and hence the LBD-FDR is optimal among false discovery rate controlling procedures. 

 Here, we focus on the canonical mean-shift Gaussian model, where
 \[X_k=\mu_k+\varepsilon_k, \quad (k=1,\dots,n)\]
for $\varepsilon_k\stackrel{\text{iid}}{\sim } N(0,1)$. Analysis of this case is standard throughout the change-point literature, \citep[e.g.][]{Verzelen2023,jang2024fastoptimalchangepointdetection}. Moreover, optimality results derived in the Gaussian case using likelihood ratio statistics can be expected to carry over to other distributional settings \citep[e.g.][]{Brown1996,Grama1998}.

\subsection{Sufficient Energy for Detection using the LBD-FDR}\label{sec:PowerNecandSuff}

In this subsection we establish our consistency result, which derives the sufficient energy, $\mathcal E_k$, that a subset of the $N_n$ change-points must satisfy in order to be detected with probability tending to 1. The precise conditions under which this is satisfied are presented in Theorem \ref{thm:power}.

\begin{theorem}[Powerfulness above a Threshold]\label{thm:power}
Suppose that some subset $S_n$ containing $K_n$ of the $N_n$ changepoints satisfies equation \eqref{eq:PowerfulEnergy}; these changepoints have ``high-energy''. Moreover, assume that there exists some $r_n\in (0,1)$ such that
\begin{equation}\label{eq:MaxSpacingAssumption}
    \max_{k}\left\{(\tau_{k}-\tau_{k-1})\vee (\tau_{k+1}-\tau_{k})\right\}\leq n^{r_n}
\end{equation}
and $(1-r_n)\log(n)\to \infty$. Finally, let $c_n$ be a fixed sequence of positive numbers, uniform in $k_i$ such that \[c_n=o\left(\sqrt{\log(n)}\right),\qquad \sqrt{\log\log(n)}+\frac{1}{\sqrt{1-r_n}}=o(c_n).\]
Then if for each changepoint $\tau_{k_i} \in S_n$
\begin{align}\label{eq:PowerfulEnergy}
\mathcal E_{k_i} \ \geq \ \sqrt{2
\log\left(\frac{n}{K_n}\right)}+\sqrt{2\log(K_n)}+c_n\quad (k_i=1,\dots, K_n),
\end{align}
then LBD-FDR detects all changepoints in $S_n$ with probability tending to $1$. Moreover, the procedure localizes each detected changepoint in $S_n$ within an interval $I\subset (\tau_{k_i-1},\tau_{k_i+1}]$.
\end{theorem}

 The proof of Theorem \ref{thm:power} can be found in Appendix \ref{sec:NecandSufEnergyLBDFDRproof}. Such an analysis for the LBD-FDR is made possible by the independent ``certificate'' set characterisation of the IndBH procedure. In particular, a non-empty independent set $C$ is called a certificate set for IndBH if $p_j\leq \alpha |C|/m_n$ for all $j \in C$, and IndBH run at level $\alpha$ rejects $H_i$ whenever there exists a certificate set $C$ containing $i$. 
 
For each triplet in a candidate certificate set, the test statistic $T_{t_k}( X) \sim Z_{t_k}$ has a normal distribution with variance $1$ and mean approximately $\mathcal E_k$. The three terms in the energy then have three different roles in the proof, $\sqrt{2\log\left(\frac{n}{K_n}\right)}$ is asymptotically the IndBH certificate threshold, and is a consequence of controlling the false discovery rate. $\sqrt{2\log(K_n)}$ is the cost of ensuring all $K_n$ target triplets are rejected simultaneously, and $c_n$ is the slowly growing slack which absorbs approximation errors and which forces the stated detection event to hold with probability tending to $1$. Moreover, examination of the proof demonstrates that the LBD-FDR constructs certificate sets of intervals covering each high energy changepoint $\tau_{k_i}$ by the shortest Bonferroni triplet that has energy at least as large as given in \eqref{eq:PowerfulEnergy}. Therefore, the set of minimal declared intervals produces the shortest regions in the data that may contain a changepoint that are still detectable by our method.

Equation \eqref{eq:MaxSpacingAssumption} is enforced to ensure that the Bonferroni triplets can suitably approximate the true energy of a changepoint. In particular, the Bonferroni triplet can only approximate the true energy, at scale $L_n$, if 
\begin{equation}\label{eq:BonfScaleNecCond}
\log\left(\frac{en}{L_n}\right)\to \infty.
\end{equation} 
A simple calculation demonstrates that equation \eqref{eq:MaxSpacingAssumption} is a sufficient condition for equation \eqref{eq:BonfScaleNecCond}. Although this generalisation degrades the rate at which $c_n\to \infty$ compared to \cite{jang2024fastoptimalchangepointdetection}, it does not affect the threshold, since at least one of $\sqrt{2\log\left(\frac{n}{K_n}\right)}$ and $\sqrt{2\log(K_n)}$ is $O\left(\sqrt{\log(n)}\right)$. Moreover, equation \eqref{eq:MaxSpacingAssumption} does not exclude many adverse changepoint arrangements, and in particular still allows sub-polynomial spacing between change-points when $K_n=n^{o(1)}$. For example, let $K_n=\log(n)+1=n^{\frac{\log(\log(n))}{\log(n)}}+1$, and suppose that all but one change-point, $\tau'$, are evenly spaced, so that the maximum distance between change-points is $n^{1-\frac{\log(\log(n))}{\log(n)}}$ (which satisfies equation \eqref{eq:MaxSpacingAssumption}). Then we can place $\tau'$ arbitrarily close to any other change-point and still maintain the maximum spacing assumption. This is important to note for the lower bound presented in the following section. 

To understand when the LBD-FDR procedure is more sensitive than Lean Bonferroni Detection requires comparing \eqref{eq:PowerfulEnergy} and equation (12) of \cite{jang2024fastoptimalchangepointdetection}, which reads 
\begin{align}\label{eq:jangenergy}
    \mathcal E_k\geq \sqrt{2\log\left(\frac{n}{(\tau_k-\tau_{k-1})\wedge (\tau_{k+1}-\tau_k)}\right)}+\sqrt{2\log(K_n)}+c_n
\end{align}
for some $c_n\to \infty$ common to all change-points arbitrarily slowly. It is clear that the second two summands in equations \eqref{eq:jangenergy} and \eqref{eq:PowerfulEnergy} are the same, so that any difference in the energy required for detection is due to the first term. This term is introduced as the data dependent threshold for controlling the false discovery rate; in particular, it ensures that the tests for each local hypothesis within a candidate set are of the same order as the rejection threshold so that the remaining two terms can control the Type II error and ensure the power tends to 1. 

It is the difference between $\sqrt{2\log\left(\frac{n}{(\tau_k-\tau_{k-1})\wedge (\tau_{k+1}-\tau_k)}\right)}$ and $\sqrt{2\log\left(\frac{n}{K_n}\right)}$, that drives whether a changepoint may be detectable by LBD-FDR, but below the minimax detection threshold for Type I error controlling methods. This is the case if
\begin{align}\label{eq:spacingvsnumb}
 (\tau_k-\tau_{k-1})\wedge (\tau_{k+1}-\tau_k)<K_n \iff  \sqrt{2\log\left(\frac{n}{K_n}\right)}<\sqrt{2\log\left(\frac{n}{(\tau_k-\tau_{k-1})\wedge (\tau_{k+1}-\tau_k)}\right)}.
\end{align}
For any $K_n$, equation \eqref{eq:spacingvsnumb} may hold when the change-points have irregular spacing, and for regular spacing, e.g.\ $K_n \asymp n\left\{(\tau_k-\tau_{k-1})\wedge (\tau_{k+1}-\tau_k)\right\}^{-1}$, equation \eqref{eq:spacingvsnumb} holds when $K_n=n^{q}$, for any $q>\frac 12$. This can be explained by a classical comparison in the multiple testing literature where false discovery rate controlling methods can choose a lower threshold than  Type I error controlling methods \citep[e.g.][]{ABRAMOVICH1996351,Abramovich2006}. Type I error controlling procedures need to protect against extreme outliers under pure noise, and so only respond to a signal that is beyond some threshold governed by the maxima of the test-statistics under the null. IndBH, analogously to the Benjamini-Hochberg procedure, has a data dependent threshold, so that many potentially weak non-null signals (i.e.\ below the threshold of the maxima of null test-statistics) can be aggregated to lower the rejection threshold. Manifestly, the sufficient detection energy from controlling the Type I error, e.g.\ a weighted Bonferroni correction, is driven by the most challenging local scale, while by controlling the false discovery rate, the method can exploit the existence of large certificate sets of moderate signals. In this way, the LBD-FDR is count adaptive, whereas the LBD-FWER is scale adaptive, and so the former method ``pays'' less to detect a growing number of moderate signals.

The results of this section are visually summarized in Figure \ref{fig:EneergyConstants}, with regions (a') and (b') being the region where the LBD-FDR can detect changepoints with a smaller energy than all Type I error controlling procedures for the case of regular and subpolynomial spacing, respectively.

\begin{figure}[!ht]
    \centering
    \includegraphics[width=0.9\linewidth]{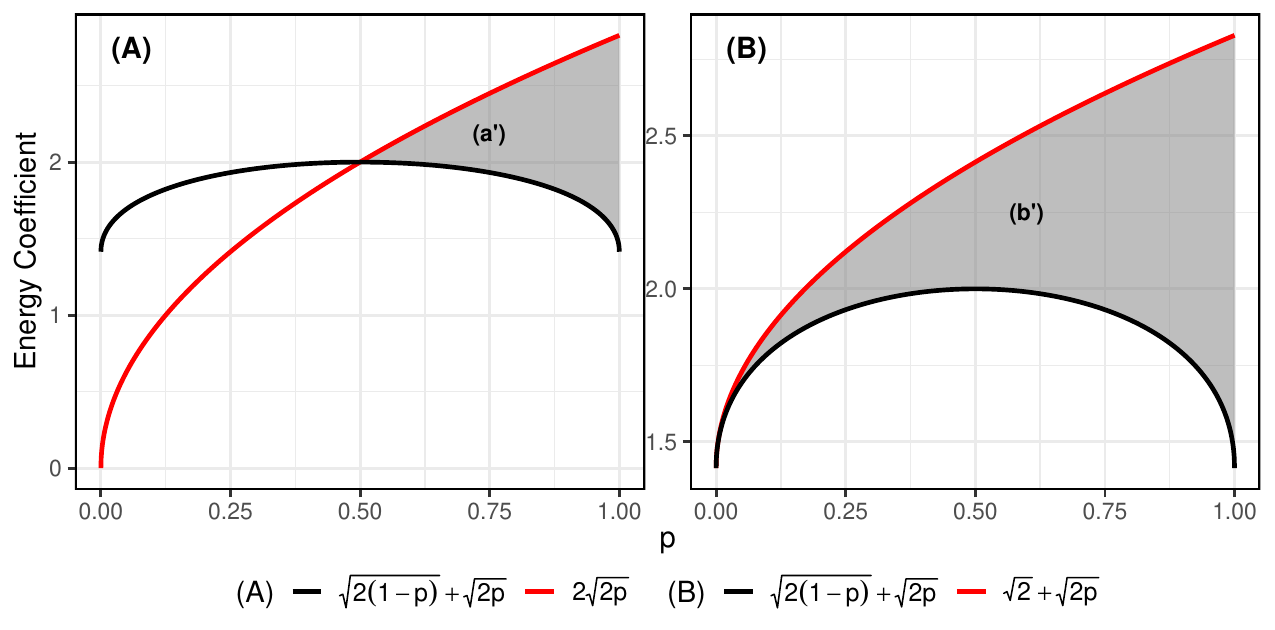}
    \caption{(A) Sufficient energy coefficient corresponding to $K_n=n^{p+o(1)}$ regularly spaced changepoints for the LBD-FDR (black) and LBD-FWER (red). (B) Sufficient energy coefficient corresponding to detection of a closely spaced changepoint, e.g. $\tau_{k+1}-\tau_k=n^{o(1)}$, when there are $K_n=n^{p+o(1)}$ changepoints for the LBD-FDR (black) and LBD-FWER (red).}
    \label{fig:EneergyConstants}
\end{figure}

\subsection{Minimax Detection Boundary}\label{sec:MinMaxImpossiblity}

This subsection is devoted to deriving the necessary changepoint energies for detection by false discovery rate controlling procedures. In particular, Theorem \ref{thm:ImpossLocalScale} states that the LBD-FDR sharply obtains the optimal constant of $2^{1/2}$ in a specific sub-class of changepoint arrangements.

Only recently have impossibility results for false discovery rate controlling procedures been developed in the multiple testing literature \citep[e.g.][]{AriasCastro2017, Rabinovich2020}, and mostly in the case of independent hypotheses. We are concerned with detection thresholds for false discovery rate controlling procedures in the multiscale setting, which requires delicate care to balance statistical indistinguishability and geometric constraints. Although similar results exist in the literature for Type I error controlling procedures \citep[e.g.][]{dumbgen2001,Dumbgen2008,ChanWalther2013,jang2024fastoptimalchangepointdetection}, as far as we know, these are the first locally sharp impossibility results for false discovery rate controlling procedures in the multiscale hypothesis test setting. In particular, we note that our result is different to that of FDRSeg, in that FDRSeg obtains optimal minimax \textit{rates} for detection, and mean estimation, up to a polylogarithmic factor \citep[see ][Section 3]{Li2016FDRSeg}. However, our optimality conditions are concerned with the leading constant for the reason explained at the end of Section~\ref{StatSetting}, see also \citep[]{Castro2005,walther2024calibrating}. For a more detailed discussion of Theorem \ref{thm:ImpossLocalScale}, and full proofs, see Appendix \ref{sec:ImpossibilityProofs}.

\begin{theorem}[Scale Dependent Detection Impossibility Theorem]\label{thm:ImpossLocalScale}
    Let $\mu_n^{(0)}$ be an arrangement of $K_n^{(0)}$ changepoints of arbitrary energy. Next, fix sequences $B_{\max,n}$, $b_n=B_{\max,n}+1$, and $\delta_n$ such that
    \[M_n=\sum_{k=0}^{K_n^{(0)}}\left\lfloor\frac{(\tau_{k+1}-\tau_k-2b_n)_+}{2\delta_n}\right \rfloor\to \infty\]
as $n\to \infty$. 
    
     Let $\psi_n(\alpha)$ be any false discovery rate level $\alpha$ test combining $p$-values computed on any family of triplets $I$. Then, there exist collections of $K_n\gg K_n^{(0)}B_{\max,n}^3\log^2(n)$ perturbative changepoints, $\mathcal T_+$, of ``small'' energy satisfying equation \eqref{eq:smallenergy}, with spacing at most $B_{\max,n}$, such that 
    \begin{align}\label{eq:TrivialPower}
        \Prob_{\mu_n}\left(\psi_n(\alpha) \text{ detects all changepoints in $\mathcal T_+$ using only intervals of length at most $B_{\max,n}$} \right)\leq \alpha+o(1)
\end{align}
where $\mu_n$ is the set of changepoints in $\mu_n^{(0)}$ perturbed by those in $\mathcal T_+$.

In particular, we define ``small'' energy changepoints as those satisfying the following equation \eqref{eq:smallenergy}

    \begin{equation}\label{eq:smallenergy}
        \mathcal E_k^2 \leq \begin{cases}
            \left(1-\varepsilon_n^{(1)}\right)^2\left[\sqrt{2\log\left(M_n\right)}-\sqrt{2\log(K_n)}\right]^2, \quad 1\leq K_n\lesssim M_n^q, \ 0\leq q<\frac{1}{4}\\
            \left(1-\varepsilon_n^{(2)}\right)^2\left[\log\left(M_n\right)-2\log(K_n)\right], \quad M_n^{1/4}\lesssim K_n\lesssim M_n^q, \ \frac{1}{4}\leq q<\frac{1}{2}, \\
        \end{cases}
    \end{equation}
    for $\varepsilon_n^{(1)}\to 0$ such that $\varepsilon_n^{(1)}\sqrt{2\log\left(M_n\right)}\to \infty$, and
    for some $\varepsilon_n^{(2)} \to 0$ such that $\varepsilon_n^{(2)}\sqrt{\log\left(M_n\right)}\to \infty.$
\end{theorem}

The interpretation of Theorem \ref{thm:ImpossLocalScale} is that the LBD-FDR obtains the optimal detection constant in certain regimes. In particular, when restricted to information on the scale $B_{\max,n}=n^{o(1)}$, if the changepoint arrangement $\mu_n^{(0)}$ is such that $\delta_n,K_n$ and $K_n^{(0)}$ are all sub-polynomial, and there are enough perturbative changepoints so that $n^{o(1)}=K_n\gg K_n^{(0)}B_{\max,n}^3\log^{2}(n)$ then the necessary and sufficient energies for detecting the changepoints $\tau \in \mathcal T_+$ are $\mathcal{E}_k = (1\mp \varepsilon_n)\sqrt{2\log(n)}$ respectively, for some $\varepsilon_n \to 0$ not much slower than $\{2\log(n)\}^{-1/2}$.

 To supplement the results of this section, we also discuss the necessary and sufficient conditions for detecting high energy segments, a generalization of high-energy changepoints introduced by \cite{Verzelen2023}, in Appendix \ref{subsec:EnergySegments}. We do so analogously to Subsections \ref{sec:PowerNecandSuff} and \ref{sec:MinMaxImpossiblity}.

\section{Improved Implementation and Numerical Experiments}\label{sec:Implementation}

In this section, we exploit the interval overlap structure of the dependency graph $\D$ induced by the sparse family of Bonferroni triplets. We first discuss how we use the interval structure of $\D$ to improve computational efficiency using dynamic programming and interval scheduling. Following this, we compare the LBD-FDR to other state of the art methods in a simulation study.

\subsection{Improved Computational Efficiency}
The IndBH algorithm, out of the box, is not suitable for multiscale statistics. IndBH is most computationally efficient when it runs on a block structure dependency graph, specifically, one in which the dependency graph $\D$ can be partitioned into many smaller connected components $\D_k$. Although this is relevant for many other applications, this may not be the case for a sufficiently rich set of multiscale intervals. Secondly, since the IndBH rejection set is a subset of the Benjamini-Hochberg rejection set, the full IndBH algorithm is run only on hypotheses first rejected by the Benjamini-Hochberg procedure. In many scenarios, especially those with frequent and large jumps, considering only those hypotheses first rejected by Benjamini-Hochberg may not reduce the computational burden much. As a result, our change-point detection scenario does not conform to situations amenable to the current implementation of IndBH.

The main computational bottleneck for the general graph IndBH procedure is computing the following quantity for $t=1,\dots, \left|\mathcal R^{(BH)}(\alpha)\right|$
\begin{equation}\label{eq:Ncompt}
    N_{k}(t)=\max_{S\in \mathcal M(\D_k)}\left|\left\{v\in S\mid \left\lceil\frac{m_np_v}{\alpha}\right \rceil \leq t\right\}\right|=\max_{S\in \mathcal M(\D_k)}|S\cap U_t|
\end{equation}
where $\mathcal M(\D)$ is the set of maximal independent sets of the subgraph $\D_k$ and $U_t=\{v\in \D_k\mid p_v\leq \alpha t/m_n\}$. The following proposition relates $N_k(t)$ to an easily computable quantity for interval graphs.
\begin{proposition}\label{prop:NcalphaDkt}
    For any graph component $\D_k$, $N_k(t)=\text{IndNum}(\D_k[U_t])$ where $\text{IndNum}(\D[U_t])$ is the maximum independent set size of subgraph $\D[U_t]$ with the vertices and edges of $U_t$.
\end{proposition}
\begin{proof}
    For any maximal independent set $S$, $S\cap U_t$ is also an independent set such that $|S\cap U_t|\leq \text{IndNum}(\D[U_t])$.

    On the other hand, consider a maximum independent set $I\in \D[U_t]$. As $I$ is independent in $\D$ it can be extended to a maximal independent set $S\supset I$ in $\D$ by greedily adding vertices. Therefore, $|S\cap U_t|\geq |I|=\text{IndNum}(\D[U_t]).$
\end{proof}

For general graphs, computing $\text{IndNum}(\D[U_t])$ may be NP-hard, so to avoid computing $\text{IndNum}(\D[U_t])$ many times, IndBH computes the maximal independent sets of $\D$ and then uses inclusion and exclusion checks by looping over all maximal independent sets to compute $N_k(t)$. Classical results in graph theory, e.g.\ \cite{MillerMuller1960} or \cite{MoonMoser1965}, have shown that there are graphs with maximal independent sets exponential in the number of vertices; hence, this loop may have exponentially many components, which is computationally infeasible. We can circumvent this problem by using dynamic programming, since our dependency graph is an interval overlap graph so that an independent set is exactly a set of pairwise disjoint intervals. For each of the connected components, this allows us to compute the entire profile $t\mapsto \text{IndNum}(\D_k[U_t])$ in one sweep over the activation times $t$, using a Fenwick tree over compressed endpoints. Pseudo-code, and further details for the specialized algorithm can be found in Appendix \ref{subsec:compdetails}. 

We next demonstrate, in Proposition \ref{prop:LBD-FDRcomplexity}, that the improved algorithm has time complexity at most $\tilde O(n^2)$. Moreover, in moderate SNR regimes (where there are $O(1)$ rejections per change-point) with independent exponential family data, the time complexity is $O\left(n\log^{7/2}(n)+K_n^2\log(n)\right)$. 

\begin{proposition}\label{prop:LBD-FDRcomplexity}
   The LBD-FDR algorithm, computed on $n$ data-points assumed to follow an exponential family distribution, has time and memory complexity
   \begin{equation*}
       T=O\left(n\log^{7/2}(n)+R^2\log(n)+P_t\right), \quad M=O\left(n\log^{5/2}(n)+n_cR+P_m\right)
   \end{equation*}
   where $R\leq n\log^{5/2}(n)$ is the size of the Benjamini-Hochberg rejection set, $n_c\leq R\leq n\log^{5/2}(n)$ is the number of connected components of the rejected subgraph and $P_t,P_m$ are the time and memory complexity of computing the $p$-values. 

\end{proposition}
The calculation can be found in Appendix \ref{subsec:proofmemorycomplex}. In the case of exponential family data, computing the $p$-values for all Bonferroni triplets has time, and memory, complexity $P_t=O(m_n)$ by precomputing the cumulative sums of $X_i,X_i^2$; see Appendix B.2 of \cite{jang2024fastoptimalchangepointdetection} for more details. As one may expect, the time complexity increases for non-parametric distributions, and using the Wilcoxon rank sum for the local tests induces a time cost of  $P_t=O\left(n^2\log^{7/2}(n)\right)$, with memory complexity at most $O(nm_n)$. Using the worst case bounds on $n_c\leq R\leq m_n$, the IndBH calculation dominates the time and memory complexity in both scenarios, producing the worst case time complexity and memory requirement of $T=O\left(n^2\log^{6}(n)\right)$, and $M=O\left(n^2\log^{5}(n)\right)$. In practice, the worst case complexities may not be achieved for two main reasons. First, except for cases with many extremely high energy change-points, $R$ will be much less than $m_n$. Additionally, as $R$ increases, the number of connected components of the Benjamini-Hochberg procedure induced subgraph generally decreases since more intervals overlap. This is exemplified in the extreme case when $R=m_n$, which then implies $n_c=1$, so that $M=O\left(n\log^{5/2}(n)+P_m\right)$. Existing methods, such as LBD-FWER, SMUCE, MQS and FDRSeg also enjoy quasi-linear time complexity in reasonable instances, although performance can degrade in adversarial scenarios to quasi-quadratic and even quasi-cubic for FDRSeg and MUSLCE.

\subsection{Empirical Power Simulations}\label{sec:simulations}

To complement our theoretical results, we empirically compare the power and the precision for a \textit{detected} changepoint, $\tau$, of the LBD-FDR to five existing methods: LBD-FWER, \citep{jang2024fastoptimalchangepointdetection} SMUCE \citep{SMUCE2014}, FDRSeg \citep{Li2016FDRSeg}, MQS \cite{Vanegas03072022} and MUSCLE \cite{liu2025multiscalequantileregressionlocal}. In particular, we compare the proportion of changepoints detected (power), and the penalised best worst case error between a changepoint's true location and the end point of any declared interval covering $\tau$ (precision). To define the penalised precision, first define 
\[p_\tau=\min_{\substack{I\in \mathcal I_c\\\tau\in I}}\max_{x\in I}|x-\tau|,\]
where $\mathcal I_c(\alpha)$ is the set of significant intervals output by LBD-FDR. Then, the penalised precision is
\[p_{\rm pen}=\frac{1}{K_n}\sum_{k=1}^{K_n}\left(p_{\tau_k}\mathbbm{1}\left\{\tau_k\text{ detected}\right\}+n\mathbbm{1}\left\{\tau_k\text{ is not detected}\right\}\right).\] 
Smaller values of $p_{\rm pen}$ indicate more accurate localisation, with missed changepoints penalised by the maximum possible error.

For the triplet based methods, the declared intervals are the minimal rejected intervals. For SMUCE and MQS, we use the reported confidence regions. For FDRSeg and MUSCLE, whose output is the set of estimated changepoints, we follow the convention that $\hat \tau_i$ is a true discovery if
\[\text{there exists } \tau \in \mathcal T \text{ such that } \tau \in \left[\frac{\lceil n(\hat \tau_{i-1}+\hat \tau_i)/2 \rceil}{n},\frac{\lceil n(\hat \tau_{i}+\hat \tau_{i+1})/2 \rceil}{n}\right),\]
with the obvious boundary conditions for the first and last estimated changepoints.

We investigate the power, precision, and empirical error (Type I error or false discovery rate) of these methods in five different settings. Namely, we consider the effect on performance when the number of evenly spaced changepoints increases for a fixed $n$, the type of additive (symmetric) errors has different tail behaviour, and in the case where the change-points have different energies. For all experiments, we average the results over 500 independent Monte Carlo trials and set $\alpha=0.10$ for all six methods. 

The purpose is to demonstrate that the LBD-FDR balances power, precision and validity more effectively than the other methods. In summary, while FDRSeg often has the greatest power, its precision is worse than the LBD-FDR, and in heavy-tailed settings its false discovery rate guarantee is no longer supported theoretically and violated empirically. Compared with MUSCLE, the quantile based analogue of FDRSeg, LBD-FDR enjoys better precision in all considered settings and also comparable or better power, especially under Gaussian noise.

For our simulations, we consider the generalized Teeth10 changepoint arrangement or mix signals \citep[e.g.][]{WBS2014} or \citep[][Appendix C]{jang2024fastoptimalchangepointdetection}. We define the Teeth$(N,n,c)$ signal as the arrangement of $N-1$ evenly spaced changepoints, in an $n$-dimensional mean-vector of jump-height $c$. The mean-vector has entries
\begin{equation}\label{eq:teethNnc}
    \mu_i=\sum_{j=1}^{N/2} c\mathbbm{1}_{\frac{n}{N}(2(j-1),2j]}(i), \quad (i=1,2,\dots,n)
\end{equation}
and we assume $N$ to be even. In our first setting, we fix $n=256$ and $N=8$, then vary $c$, equivalently the signal to noise ratio, between $0.1$ and $3.1$. In our second, we fix $n=256$, $N=128$ and vary $c$ between $1.1$ and $5.1$. In our third setting, we keep the same fixed signal, the Teeth$(8,256,2)$ arrangement, and vary the distribution of the additive noise. In particular we select $t(\nu)$ with degrees of freedom $\nu \in \{1,2,5,10,20\}$. In our fourth and fifth settings we consider the mix signal with either $N(0,\sigma^2)$ noise for $\sigma\in \{2,4,6,8\}$ or $t(\nu)$ noise for $\nu\in\{1,2,3,4\}$.  When the errors follow a Gaussian distribution with known variance we use the CUSUM statistic for the triplet methods, and the known variance setting for SMUCE and FDRSeg. When the noise has a $t$-distribution we use two sample Kolmogorov-Smirnov test statistics for the triplet based methods, and use SMUCE and FDRSeg with the unknown variance setting. In all simulations, MQS and MUSCLE are performing median regression. Power and precision are visually summarized in Figure \ref{fig:prec&power} in the teeth experiments, and in Figure \ref{fig:mix} for the mix signal experiements. 

\begin{figure}[!ht]
    \centering
    \includegraphics[width=\linewidth]{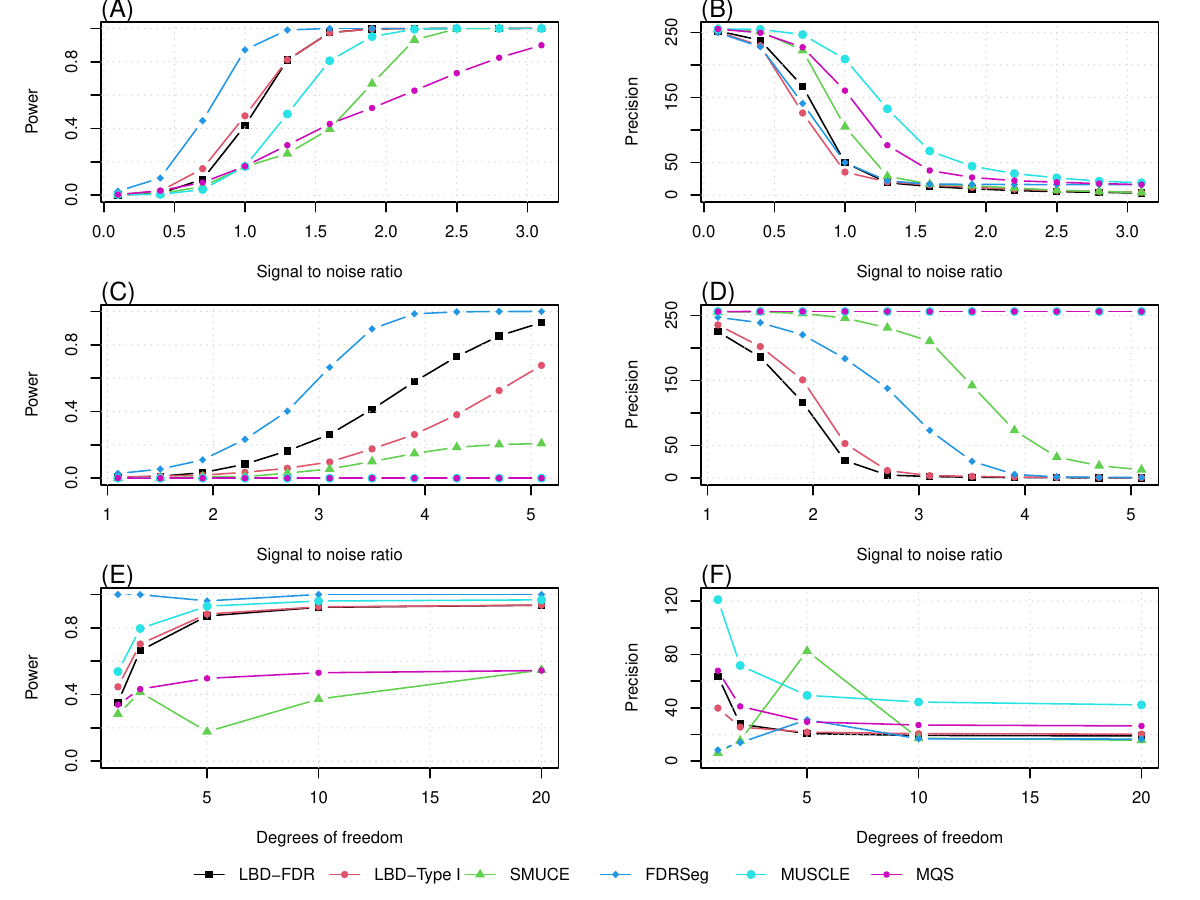}
    \caption{For all of the above plots the methods are coloured ``LBD-FDR'' (black), ``LBD-FWER'' (red), ``SMUCE" (green), ``FDRSeg'' (blue), ``MUSCLE'' (cyan), ``MQS'' magenta. (A) Power vs. signal to noise ratio for the Teeth$(8,256,c)$ for $c\in [0.1,3.1]$ with iid $N(0,1)$ additive noise. (B)  Precision vs. signal to noise ratio for the same simulated trials as (A). (C) Power vs. signal to noise ratio for the Teeth$(128,256,c)$ for $c\in [1.1,5.1]$ with iid $N(0,1)$. (D)  Precision vs. signal to noise ratio for the same simulated trials as (C). (E) Power vs. degrees of freedom, $\nu$, for the Teeth$(8,256,2)$ signal with iid $t(\nu)$ additive noise. (F) Precision vs. degrees of freedom for the same simulated trials as (E) The upper bound on precision is $n=256$, and a smaller precision is better.}
    \label{fig:prec&power}
\end{figure}

\begin{figure}
    \centering
    \includegraphics[width=\linewidth]{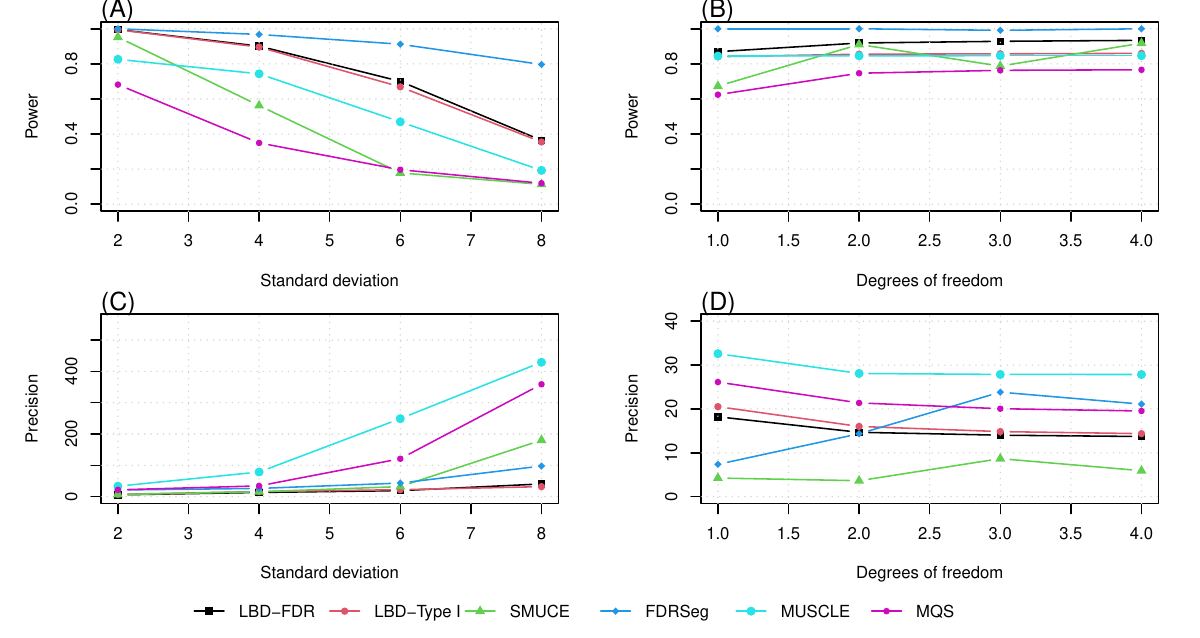}
    \caption{The colouring convention matches \ref{fig:prec&power}. (A) Empirical power from the ``mix'' signal with $N(0,\sigma^2)$ errors for $\sigma \in \{2,4,6,8\}$. (B) Empirical power from the ``mix'' signal with $t(\nu)$ errors for $\nu \in \{1,2,3,4\}$. (C) Empirical penalised precision for the same experiment as (A). (D) Empirical penalised precision for the same experiment as (B).}
    \label{fig:mix}
\end{figure}

As expected, when the signal to noise ratio increases, both power and precision tend to improve. In the case where there are 7 change-points, e.g.\ Figures \ref{fig:prec&power} (A) and (B) in the greatest signal to noise ratio regimes, both triplet based methods, FDRSeg and MUSLCE detect most, if not all, changepoints and localize them precisely. Although the power of FDRSeg is greater than the Bonferroni triplet based methods, its localisation accuracy is slightly worse. In the case when there are 127 change-points (Figures \ref{fig:prec&power} (C) and (D)), which is the closest spacing possible, the LBD-FDR and FDRSeg have greater power than LBD-FWER and SMUCE, while MQS and MUSCLE fail to detect any changepoint. Despite this, when a changepoint is detected, the two triplet based methods have comparable precision, both of which have better precision than SMUCE and FDRSeg. In all cases, SMUCE and MQS appear to be conservative; this is expected and was the motivation behind the development of FDRSeg and MUSCLE respectively. For the mix signal (Figure \ref{fig:mix}) the story is essentially the same: Averaged over the changepoints, FDRSeg enjoys the greatest power but has worse penalised precision compared to the triplet based methods.

In Figures \ref{fig:prec&power} (E) and (F), corresponding to the case of the teeth signal and $t$-distributed noise, it is clear that FDRSeg is the most powerful, and has comparable precision to the other methods, that is, the precision of the triplet methods has worsened. This is to be expected since non-parametric tests are expected to be conservative when compared to methods that make the correct assumption about the data's distribution. While MUSCLE has the second best power after FDRSeg, its precision is the worst. In all settings, the triplet based methods, MUSCLE and MQS empirically control their respective error rates, as is expected theoretically. On the other hand, SMUCE and FDRSeg do not empirically control the Type I error or false discovery rates, respectively, when the noise does not have a second moment (in all cases the estimated error rates exceed $5\alpha=0.5$). In the other settings SMUCE and FDRSeg empirically control their error rates. 

The trade off between validity and power or precision should be apparent. When assuming less about the data by using a non-parametric method, a practitioner ensures validity in a wider class of models at the cost of a loss in power and/or precision when compared to methods that correctly specify the data's distribution. We stress that validity in heavy-tailed or non-parametric settings is needed, since heavy tails, or non-Gaussian data, can be expected in several applications such as array comparative genomic hybridization \citep{Huang2007}, or well-log data \citep[e.g.][]{Fearnhead02012019,oruanidh1996}. Overall, while the LBD-FDR, LBD-FWER, MQS and MUSCLE are always valid, the LBD-FDR seems to enjoy comparable, or better, power and/or precision than the other three valid methods in the tested settings.

\section{Discussion}\label{sec:discussion}

We suggest three directions for future research. The first is to establish more refined methods and proofs in the known variance Gaussian setting. First of all, we do not derive an explicit detection threshold when there is a polynomial number of changepoints, or when the spacing is large. It is challenging to develop impossibility results, over all arrangements, for false discovery rate controlling procedures due to the interplay between the data adaptive nature of these methods and the rich geometry of all changepoint arrangements. Only recently have impossibility results for false discovery rate controlling procedures been developed in the multiple testing literature \citep[e.g.][]{AriasCastro2017, Rabinovich2020,Castillo2020,Roquian2022}, and most studies focus on the case with independent hypotheses. The multiscale nature of changepoint detection inherently creates geometric obstacles in the way of significant correlation structures and a loss of bijectivity between a single change-point and one rejected hypothesis. 

Despite the difficulty of establishing optimality, graph adapted procedures are natural for multiscale detection problems in any finite dimension. Due to its analytic tractability, IndBH is a strong candidate as the base for many multiscale procedures, however, it is not the most liberal graph-adapted procedure. More liberal graph adapted procedures are defined as the fixed point of a sequential algorithm where early stopping breaks validity, or potentially exponentially many recursive calls to IndBH, both of which come with a large computational burden. Moreover, it is not clear whether deterministic pruning of their rejection sets still results in a false discovery rate controlling method. We leave the implementation and theory of those more liberal procedure open.

Our final suggested direction for future research concerns modifications so that the LBD-FDR is valid in the setting when there are heteroscedastic or weakly correlated errors. The false discovery rate guarantee of our method explicitly relies on the data sequence being independent, and our theoretical guarantees fail otherwise. Moreover, heteroscedastic errors break the assumed exchangeability of the data, under the null, over all triplets within the Bonferroni triplet family. 

\subsection{Declarations (Reproducibility and Use of generative AI)}
All code and simulation data used for this study are freely available at \url{https://github.com/davlo199/LBD-FDR}. 

During the preparation of this work the authors used GPT-5.2 Pro to assist in the development of the code, aid in calculations, as well as finding suitable references cited within the work. After using this tool/service the authors reviewed and edited the content as necessary and take full responsibility for the content of this article.

\subsection{Acknowledgements}
L. Davis would like to thank Yash Nair and Conor Kresin for their helpful conversations and feedback. G. Walther acknowledges support from NSF grant DMS-2413885.

\bibliographystyle{abbrvnat}
\bibliography{bibliography}
\clearpage

\appendix

\section{Proof of Theorems \ref{thm:prunedFDRcontrol} and \ref{thm:power}}

In this first section of the appendix we present the remaining details of Theorem \ref{thm:prunedFDRcontrol} and Theorem \ref{thm:power}. 

The first Subsection \ref{sec:Pruningproof} contains a short argument that exploits the representation of independent sets of hypotheses as disjoint sets of intervals. Indeed, for a triplet $t$ to be rejected, it must be part of a certificate set of size $|\beta_t|$. But the independence structure of the hypotheses corresponds exactly to that of intervals which disjoint and hence minimal with respect to inclusion. Therefore, pruning the IndBH rejection set does not remove too many non-null rejections. 

The second Subsection \ref{sec:NecandSufEnergyLBDFDRproof}, proofs Theorem \ref{thm:power}. This proof follows fairly easily from Theorem 3.1 of \cite{jang2024fastoptimalchangepointdetection} and the certificate set structure of IndBH. 

\subsection{Proof of Theorem \ref{thm:prunedFDRcontrol}}\label{sec:Pruningproof}

\begin{proof}

    For any null triplet $ t =(s,m,e)$ define $I( t)=(s,e]$, i.e.\ its outer interval. Equations (2)-(4) of \cite{nguyen2025controllingfalsediscoveryrate} imply that 
    \[ t \in \mathcal R_\alpha(\text{all})\iff p_t\leq \frac{\alpha \beta_{ t}}{m_n}\]
    where 
    \[\beta_{ t}=\max\left\{r\geq 1\ \mid \ \text{ there exists } A \subset N_{\D,  t}^c \text{ independent }, |A|=r-1, \ p_u\leq \frac{\alpha r}{m_n}\ \text{ for every } u \in A\right\}\]
    so that $\beta_{ t}$ is the size of the largest independent certificate set of the IndBH procedure. It is clear that $\beta_{ t}$ is only a function of the $p$-values that are independent of $p_{ t}$.

    Now let $\square \in \left\{\min,\text{disj}\right\}$ and consider some $J \in \mathcal R_\alpha(\square)$, so there exists a triplet, $ t$, such that $I( t)=J$, with corresponding certificate set $C_t$, with size $\beta_t$. Each triplet in $C_t$ is rejected and, by independence, is disjoint from every other triplet in the certificate set, since the data sequence contains independent entries. Therefore, each triplet in $C_t$ is a minimal and disjoint rejected interval, so that $|\mathcal R_\alpha(\square)|\geq \beta_t.$
    Therefore,
    \begin{equation}\label{eq:1JFDP}
        \frac{\mathbbm{1}\{J\in \mathcal R_\alpha(\square)\}}{|\mathcal R_\alpha(\square)|\vee 1}\leq \sum_{ t: I( t)=J}\frac{\mathbbm{1}\{p_{ t}\leq \alpha \beta_{ t}/m_n\}}{\beta_t}.
    \end{equation}
    Summing equation \eqref{eq:1JFDP} over $J\in \mathcal H_0$, and then taking an expectation implies,
    \begin{align}\label{eq:2JFDP}
        \text{FDR}^{\square}&\leq \E\left[\sum_{J \in \mathcal{H}_0}\sum_{ t: I( t)=J}\frac{\mathbbm{1}\{p_{ t}\leq \alpha \beta_{ t}/m_n\}}{\beta_t\vee 1}\right]\\\label{eq:3JFDP}
        &=\sum_{ t \in \mathcal H_0}\E\left[\frac{\mathbbm{1}\{p_{ t}\leq \alpha \beta_{ t}/m_n\}}{\beta_t\vee 1}\right]\\\label{eq:4JFDP}
        &=\sum_{ t \in \mathcal H_0}\E\left[\E\left[\frac{\mathbbm{1}\{p_{ t}\leq \alpha \beta_{ t}/m_n\}}{\beta_t\vee 1} \ \bigg| \ S_{ t}\right]  \right]\\\label{eq:5JFDP}
        &\leq \sum_{ t \in \mathcal H_0}\E\left[\frac{\alpha \beta_{ t}}{(\beta_{ t}\vee 1)m_n}\right]\leq\frac{\alpha |\mathcal H_0|}{m_n}.
    \end{align}
    Equation \eqref{eq:3JFDP} is a reindexing; a triplet $ t$ is null if, and only if, its outer interval $I( t)$ does not contain a change-point. Equation \eqref{eq:4JFDP} is an application of the tower property, conditioning on $S_{ t}$, which is the $\sigma$-algebra generated by all $p$-values that are independent of $p_{ t}$. Equation \eqref{eq:5JFDP} follows from the fact that $p_{ t}\perp S_{ t}$, $\beta_t$ is measurable with respect to $S_{ t}$, and $p_{ t}$ is super-uniform under the null-hypothesis.
    
\end{proof}

\subsection{Proof of Theorem \ref{thm:power}}\label{sec:NecandSufEnergyLBDFDRproof}
\begin{proof}

We only consider the case in which there are at least two detectable change-points, specifically that $K_n\geq 2$. Consider the collection of high-energy change-points $\{\tau_{k_i}\}_{i=1}^{K_n}$, we want to find a Bonferroni triplet such that the mean of the test-statistic computed on the triplet approximates the continuous energy $\mathcal E_{k_i}$ of the change-point $\tau_{k_i}$. We index the high-energy change-points as a subsequence of length $K_n$, so that our analysis is local, e.g.\ that small undetectable change-points do not enter into the certificate sets.

The proof constructs two certificate sets 
\[C^{(0)}:=\{i \text{ even}\}, \quad C^{(1)}:=\{i \text{ odd}\}\]
so that for any $i \in [K_n]$ the certificate for $\tau_{k_i}$ is $C^{(i\mod 2)}$, hence $|C^{(k)}|\asymp K_n/2$ for $k=0,1$. If the triplets $I_j,I_k \in C^{(0)}$ then necessarily $j=k$ or $I_j\cap I_k =\emptyset$ since there is a ``buffer" of at least one odd numbered change-point between the two nearest points of the interval. We thus reject the null-hypothesis that there is no change-point in the candidate triplet $I_i\subset (\tau_{k_i-1},\tau_{k_i+1}]$, if 
\begin{equation}\label{eq:pt}
    \max_{k\in \{0,1\}}\max_{i \in C^{(k)}}p_{I_i} \leq \frac{\alpha \lfloor K_n/2\rfloor}{m_n}
\end{equation}
by the certificate set interpretation of the IndBH procedure (see Section 3.1 of \cite{nguyen2025controllingfalsediscoveryrate}). 

We now construct these two certificate sets and demonstrate that equation \eqref{eq:pt} holds with probability tending to 1. The following is analogous to the first half of the proof of Theorem 3.1 in \cite{jang2024fastoptimalchangepointdetection} which is concerned only with the geometric properties of Bonferroni triplets.

Let $\tau_{k_i}$ be a high-energy change-point that satisfies equation \eqref{eq:PowerfulEnergy}, then set $a_i=\max(\tau_{k_i}-\delta_i,\tau_{k_i-1})$ and $z_i=\min(\tau_{k_i}+\delta_i,\tau_{k_i+1})$ where $\delta_i$ is the smallest non-negative integer such that 
\begin{equation}\label{eq:ApproximatingEnergy}
    \left|\mu_{\tau_{k_i+1}}-\mu_{\tau_{k_i}} \right|\sqrt{\frac{(\tau_{k_i}-a_i)(z_i-\tau_{k_i})}{z_i-a_i}}\ \geq \ \sqrt{2\log\left(\frac{n}{K_n}\right)}+\sqrt{2\log(K_n)}+c_n.
\end{equation}
Such a $\delta \in \left[1,\left((\tau_k-\tau_{k-1})\vee (\tau_{k+1}-\tau_k)\right)\right]\subset[1,n^{r_n}]$ exists since the right hand side of inequality \eqref{eq:ApproximatingEnergy} is independent of $k$. As argued in \cite{jang2024fastoptimalchangepointdetection}, there is no loss of generality in assuming $\tau_{k_i}-a_i\leq z_i-\tau_{k_i}$. It is a well known result that $z_u \leq \sqrt{2\log(1/u)}$ for $u \in (0,1)$. Therefore, for any fixed $\alpha \in (0,1)$, as $\alpha \lfloor K_n/2\rfloor/(2m_n)\geq \frac{\alpha K_n}{6m_n}$
\begin{align*}
    z_{\alpha \lfloor K_n/2\rfloor/(2m_n)}&\leq \sqrt{2\log\left(\frac{6m_n}{\alpha K_n}\right)}\\
    &\leq \sqrt{2\log\left(\frac{n}{K_n}\right)+2C\log\log(n)}\\
    &\leq \sqrt{2\log\left(\frac{n}{K_n}\right)}+C_2\sqrt{\log\log(n)}\\
    &=\sqrt{2\log\left(\frac{n}{K_n}\right)}+o(c_n),
\end{align*}
since $\sqrt{\log\log(n)}=o(c_n)$. Therefore, equation \eqref{eq:ApproximatingEnergy} can be replaced with 
\begin{equation}\label{eq:ApproximatingEnergywithZ}
    \left|\mu_{\tau_{k_i+1}}-\mu_{\tau_{k_i}} \right|\sqrt{\frac{(\tau_{k_i}-a_i)(z_i-\tau_{k_i})}{z_i-a_i}}\ \geq \ z_{\alpha \lfloor K_n/2\rfloor/(2m_n)}+\sqrt{2\log(K_n)}+\frac{c_n}{2}
\end{equation}
for all sufficiently large $n$.

The work on page 25 of \cite{jang2024fastoptimalchangepointdetection} is unchanged and we can find that the test-statistic $T_{t_k}(X)$ satisfies
\begin{align}\nonumber
    T_{t_k}( X)&\geq \left|\mu_{\tau_{k_i+1}}-\mu_{\tau_{k_i}} \right|\sqrt{\frac{(\tau_{k_i}-a_i)(z_i-\tau_{k_i})}{z_i-a_i}}\frac{\left[1-\frac{8}{\sqrt{2(1-r_n)\log(n)}}\right]^{3/2}}{\left[1+\frac{16}{\sqrt{2(1-r_n)\log(n)}}\right]^{1/2}}-|N(0,1)|\\ \nonumber
    &\geq  \left[z_{\alpha \lfloor K_n/2\rfloor/(2m_n)}+\sqrt{2\log(K_n)}+\frac{c_n}{2}\right]\left[1-\frac{24}{\sqrt{2(1-r_n)\log(n)}}\right]-|N(0,1)|\\\label{eq:z_statTtineq}
    &\geq  z_{\alpha \lfloor K_n/2\rfloor/(2m_n)}+\sqrt{2\log(K_n)}+\frac{c_n}{4}-|N(0,1)|.
\end{align}
where the final inequality holds for all $n$ sufficiently large (dependent only on $c_n$), because 
\begin{align*}
    24\frac{z_{\alpha \lfloor K_n/2\rfloor/(2m_n)}+\sqrt{2\log(K_n)}+c_n/2}{\sqrt{2(1-r_n)\log(n)}}&\lesssim 24\frac{\sqrt{2\log\left(\frac{n}{K_n}\right)}+\sqrt{2\log(K_n)}+c_n/2}{\sqrt{2(1-r_n)\log(n)}}\\
    &\lesssim \frac{\sqrt{\log(n)}+c_n}{\sqrt{\log(n)(1-r_n)}}\\
    &= \frac{1}{\sqrt{1-r_n}}+o(c_n)=o(c_n),
\end{align*}
By the Gaussian tail bounds, for $p$-values corresponding to the triplet $t_k$,
\begin{align*}
    \Prob\left(p_k \leq \frac{\alpha \lfloor K_n/2\rfloor}{m_n}\right) &= \Prob\left(1-\Phi(T_{t_k}) \leq \frac{\alpha \lfloor K_n/2\rfloor}{2m_n}\right)\\
    &=\Prob\left\{z_{\alpha\lfloor K_n/2\rfloor/2m_n}\leq T_{t_k}( X) \right\}\\
    &\geq \Prob_0\left[|N(0,1)|\leq \sqrt{2\log(K_n)}+\frac{c_n}{4} \right]\\
    &= 1-\frac{2^{1/2}}{\pi^{1/2}\left[\sqrt{2\log(K_n)}+\frac{c_n}{4}\right]}\exp\left(-\frac{\left[\sqrt{2\log(K_n)}+\frac{c_n}{4}\right]^2}{2}\right)
\end{align*}
Now consider the event that all high-energy change-points are detected, and let $b_n=\sqrt{2\log(K_n)}+\frac{c_n}{4}$.
\begin{align*}
    \Prob(\text{all $K_n$ change-points are detected})&\geq \Prob(C^{(0)}\text{ and }C^{(1)}\text{ are rejected})\\
    &\geq 1-\Prob(C^{(0)}\text{ not rejected}) -\Prob(C^{(1)}\text{ not rejected}) \\
    &=\Prob\left[\bigcap_{j\in C^{(0)}}\left\{p_j \leq \frac{\alpha \lfloor K_n/2\rfloor}{m_n}\right\} \right]-\Prob\left[\bigcup_{j\in C^{(1)}}\left\{p_j > \frac{\alpha \lfloor K_n/2\rfloor}{m_n}\right\} \right]\\
    &=\prod_{j\in C^{(0)}}\Prob\left(p_j \leq \frac{\alpha \lfloor K_n/2\rfloor}{m_n}\right)-\sum_{j\in C^{(1)}}\Prob\left(p_j > \frac{\alpha \lfloor K_n/2\rfloor}{m_n}\right)\\
    &\geq \left(1-\frac{2}{(2\pi)^{1/2}b_n}e^{-\frac{b_n^2}{2}} \right)^{\lfloor K_n/2\rfloor}-\frac{2\lfloor K_n/2\rfloor}{(2\pi)^{1/2}b_n}e^{-\frac{b_n^2}{2}} \to 1
\end{align*}
if $\frac{K_ne^{-\frac{b_n^2}{2}}}{b_n(2\pi)^{1/2}}\to 0$; by absolute convergence of the infinite product as $K_n \to \infty$. A sufficient condition is $2\log K_n\leq b_n^2$, which is clearly true for the choice of $b_n=\sqrt{2\log(K_n)}+\frac{c_n}{4}$.

\end{proof}

\section{Supplement to Section \ref{sec:MinMaxImpossiblity}}\label{sec:ImpossibilityProofs}
This appendix contains additional details relevant to Section \ref{sec:MinMaxImpossiblity}. In particular, it contains proofs as well as additional discussion relevant to how refined our arguments are and what obstacles still lie before us in regards to obtaining a sharp lower bound over change-point arrangements for any $K_n$.

The proof of Theorem \ref{thm:ImpossLocalScale} constructs a hard subfamily of change-point arrangements over which the uniform mixture distribution is statistically indistinguishable from a non-flat background. Specifically, there is an arrangement of $K_n$ change-points, ``far-away'' from the background arrangement $\mu_n^{(0)}$, that induce an isotropic Gaussian random variable that is statistically close to $N(\mu_n^{(0)},I_n)$ (e.g.\ the total variation distance between the distributions tends to 0 as $n\to \infty$). We construct such a distribution by discretising $\{1,2,\dots,n\}$ into flat segments between change-points, and uniformly at random place $K_n$ blocks to ``activate'', or place a change-point in. After combining blockwise, under $H_0$, we observe a collection of $M_n$ (defined in the theorem) independent and identically distributed $N(0,1)$ random variables. Under $H_1$ we observe $K_n$ $N(\mathcal E_k,1)$ random variables, for some positive mean $\mathcal E_K$ with the rest being null. Asymptotically, this is equivalent to 
\begin{align}\label{eq:sparsemeans}
    H_0: \ X_i \stackrel{\text{iid}}{\sim} N(0,1), \quad H_1 \ : \ X_i \stackrel{\text{iid}}{\sim}(1-\epsilon_n)N(0,1)+\epsilon_n N(\mathcal E,1)
\end{align}
with $\epsilon_n=\frac{K_n}{M_n}$, $M_n$ total hypotheses  \citep{kou2022largescaleinferenceblockstructure}.  

The problem of detecting $H_0$ vs. $H_1$ in equation \eqref{eq:sparsemeans} is exactly that of detecting a sparse mixture of normal random variables with some having positive mean; a classical problem studied by \cite{Ingster1999} and \cite{DonohoHC2004}. The necessary energy in equation \eqref{eq:smallenergy} is the same as that of \cite{Ingster1999,DonohoHC2004}, up to the reparameterisation discussed, so that the optimal procedure for detecting the existence of any change-point is the higher criticism statistic. In particular, this implies that our impossibility argument against the global null cannot be improved. As such, any further refinement in establishing the optimality for false discovery rate controlling change-point detection procedures must be done locally, e.g.\ against a non-zero background. 

We now consider local refinements against a non-zero background. Here, the lower bounds for false discovery rate and Type I error separate, and it is due to how the power and level of these tests are defined. For a Type I error controlling procedure $\phi(X)$, the power, $\beta$, and level, $\alpha$ are defined as $\E_{P_1}[\phi(X)]=\beta$ and $ \E_{P_0}[\phi(X)]\leq \alpha$ \citep[e.g.][]{Lehmann2022}, so that establishing a detection threshold is simply a consequence of the likelihood ratio converging to 1 in $L^1$. However, the false discovery rate is defined as the expectation of a ratio of two random variables so that a tractable comparison between the power and level requires analysis of the local geometry of the change-point arrangements and how these interact with the interval systems generating the sets of hypotheses.

To facilitate such a local analysis we now allow there to be a (slowly) growing number of change-points which may have arbitrarily high energy. The reason we can only consider a slowly growing number of change-points for our local analysis is that if the number of local hypotheses which are non-null is too large relative to the number of nulls, relating the false discovery rate and power under a local alternative will yield trivial bounds; this is due to the multiplicity burden of a single change-point. For most systems of intervals or triplets, a single change-point can activcate $\Omega(n)$ hypotheses. A simple construction for a non-sparse family of intervals is for $\tau=\lfloor n/2\rfloor$ and 
\[\mathcal S=\left\{(\tau-2,\tau+1],(\tau-2,\tau+2],\dots, (\tau-2,n]\right\}.\]
Clearly,  $|\mathcal S|=n$ so that a single adversarially placed change-point can induce $\Omega(n)$ non-null hypotheses. For Bonferroni triplets, a similar result holds because the longest extension length is $O(n)$. 

For a general non-zero background with a growing number of change-points, we are more restricted in our impossibility analysis as to when we can or cannot localize the perturbative change-points. Heuristically, if there are many more change-points under the background $\mu_n^{(0)}$ than we add under the local alternative $\mu_n$ the rejections corresponding to these background change-points will dominate the false discovery proportion and so it is possible to reject every interval containing only these new change-points while still controlling the false discovery rate. This is best exemplified in the following lemma. 

\begin{lemma}\label{lem:FDRpowBound}
Define $N_{\psi_n(\alpha)}(I)$ to be the number of null hypotheses that are rejected by the FDR level $\alpha$ procedure $\psi_n$. Moreover, define $m_{1,n}$ as the number of non-null hypotheses. Then,
\begin{align}\label{eq:FDRPpowUPbound}
\Prob\left(N_{\psi_n(\alpha)}(I)\geq x\right)\leq \alpha \left(1+\frac{m_{1,n}}{x}\right).
   \end{align}
\end{lemma}
\begin{proof}
The desired inequality follows from a simple calculation
 \begin{align}\label{eq:FDPmanip1}
       \text{FDP}_n&\geq  \mathbbm{1}\left\{N_{\psi_n(\alpha)}(I)\geq x\right\}\frac{N_{\psi_n(\alpha)}(I)}{R_n\vee 1}\\
       &\geq \mathbbm{1}\left\{N_{\psi_n(\alpha)}(I)\geq x\right\}\frac{N_{\psi_n(\alpha)}(I)}{N_{\psi_n(\alpha)}(I)+m_{1,n}}\\\label{eq:FDPmanip2}
       &\geq \mathbbm{1}\left\{N_{\psi_n(\alpha)}(I)\geq x\right\}\frac{x}{x+m_{1,n}}\\\nonumber
       \implies  \alpha &\geq \frac{x}{x+m_{1,n}}\Prob\left(N_{\psi_n(\alpha)}(I)\geq x\right)\iff  \Prob\left(N_{\psi_n(\alpha)}(I)\geq x\right)\leq \alpha \left(1+\frac{m_{1,n}}{x}\right).
   \end{align}
\end{proof}

In general if $x=o(m_{1,n})$ the upper bound in equation \eqref{eq:FDRPpowUPbound} becomes trivial. In our change-point detection context it is challenging to derive non-trivial upper bounds on $m_{1,n}$ over all length scales. This is because any number of additional change-points can induce $\Omega(n)$ non-null hypotheses. However, we can control $m_{1,n}$ non-trivially if we restrict our attention to hypotheses of a maximum length scale $B_{\max,n}$. When we do so, we are able to leverage equations \eqref{eq:FDPmanip1}-\eqref{eq:FDPmanip2} to derive a scale-dependent impossibility theorem against a non-zero background. This theorem is presented in the main text (Theorem \ref{thm:ImpossLocalScale}).

We now present the more technical proofs associated with this section.
\input{thm3updatedproof}

\subsection{Proof of Theorem \ref{thm:ImpossLocalScale}}

Arguing as we did in equations \eqref{eq:FDPmanip1}-\eqref{eq:FDPmanip2} restricted to the hypotheses of scale at most $B_{\max,n}$ and by applying the result of Proposition \ref{prop:TVproposition}
\begin{align*}
 \inf_{\mu_n}\Prob_{\mu_n}\left(N_{\psi_n(\alpha)}(I)\geq x\right)&\leq \E_{\mu_n\sim \pi_n}\left\{\Prob_{\mu_n}\left(N_{\psi_n(\alpha)}(I)\geq x \right) \right\}\\ \nonumber   
    &\leq \Prob_{\mu_n^{(0)}}\left(N_{\psi_n(\alpha)}(I)\geq x\right)+\norm{\Prob_{\mu_n^{(0)}}-\bar P_{\mu_n}}_{TV}\\
    &\leq \alpha \left(1+\frac{m_{1,n}}{x}\right)+\norm{\Prob_{\mu_n^{(0)}}-\bar P_{\mu_n}}_{TV}\\
    &=\alpha \left(1+\frac{m_{1,n}}{x}\right)+o(1)
\end{align*}
The remainder of the proof is devoted to deriving sufficient conditions on $x$, dependent on $K_{n}^{(0)}$ and $B_{\max,n}$, so that $\left(1+\frac{m_{1,n}}{x}\right)=1+o(1)$. To do so, consider one of the $K_n^{(0)}$ change-points, under the background, $\tau_k$. We will upper bound $c(\tau_k;B_{\max,n})=\#\{(t_1,t_2,t_3)|t_1< \tau\leq t_3 \cap t_3-t_1\leq B_{\max,n}\}$ by decomposing 
\[c(\tau_k;B_{\max,n})=c_B(\tau_k;B_{\max,n})+c_E(\tau_k;B_{\max,n})\]
where $c_B(\tau_k;B_{\max,n})$ is the largest number of Bonferroni intervals $\tau_k$ can fall into, and $c_E(\tau_k;B_{\max,n})$ is the largest number of extension sides that $\tau_k$ can fall into.

We first consider $c_B(\tau_k;B_{\max,n})$ and fix some scale $\ell$. Recalling the definition of Bonferroni intervals \citep{walther2010,WaltherPerry2021}, for the readers convenience, for any integer $\ell \in \{0,\dots,\ell_{\max}=\lfloor \log_2(n/4)\rfloor-1\}$, intervals of length $L_I\in [2^\ell,2^{\ell+1})$ are approximated by the collection of intervals
\begin{equation}
    \mathcal J_\ell=\left\{(j,k]\ \mid \ j,k\in \{id_\ell,i=0,\dots\},2^\ell \leq k-j<2^{\ell+1} \right\}
\end{equation}
where the grid spacing is $d_\ell=\left\lceil 2^\ell\left\{2\log\left(\frac{en}{2^\ell} \right) \right\}^{-1/2}\right\rceil$. The collection $\bigcup_\ell \mathcal J_\ell$ is called Bonferroni intervals.

At the scale $\ell$, $c_B^{\ell}(\tau_k)\lesssim\left(\frac{2^\ell}{d_\ell}\right)^2$, where the number of Bonferroni intervals that $\tau_k$ can fall into is denoted $c_B^{\ell}(\tau_k)$. The reason being that the left end-point of the Bonferroni interval is placed on a grid-point, $id_\ell$, such that $0< \tau-id_\ell\leq 2^\ell$ and there are at most $\frac{2^\ell}{d_\ell}$ integers $i$ satisfying this inequality. By symmetry, this holds for right endpoints as well, so multiplying the two gives an upper bound on $c_B^{\ell}(\tau_k)$. Summing over all permissible scales implies
\begin{align*}
    c_B(\tau_k;B_{\max,n}) &=\sum_{\substack{\ell\leq L \\ L\leq \log_2(B_{\max,n}) }}c_B^{\ell}(\tau_k)\\
    &\lesssim \sum_{\substack{\ell\leq L \\ L\leq \log_2(B_{\max,n}) }} \left(\frac{2^\ell}{d_\ell}\right)^2\\
    &\lesssim \sum_{\substack{\ell\leq L \\L\leq \log_2(B_{\max,n}) }} 2\log\left(\frac{en}{2^\ell}\right)=O\left(\log(n)\log(B_{\max,n})\right).
\end{align*}

We now consider $c_E(\tau_k;B_{\max,n})$ and, once again, first fix a scale $\ell$. We consider right extensions, and by symmetry this holds for left extensions too. If $\tau_k$ is in the extension, then $t_2<\tau\leq t_2+L \implies t_2\in [\tau_k-L,\tau)$. For any fixed $t_2$ there are at most $\frac{2^\ell}{d_\ell}$ Bonferroni intervals of the form $(t_1,t_2]$. Moreover, there are no more than $\frac{2L}{d_\ell}$ possible positions for $t_2 \in [\tau_k-L,\tau)$. Thus,
\[\#\{\text{left-Bonferroni triplet at scale $\ell$ with extension $L$ covering $\tau_k$}\}\lesssim \frac{L2^\ell}{d_\ell^2}.\]
Summing $L$ over all lengths, $L\leq B_{\max,n}$ implies
\begin{align*}
    \sum_{\substack{L\leq B_{\max,n}}}L\asymp\sum_{r:2^r\leq B_{\max,n}}2^r\sqrt{2\log\left(\frac{en}{2^r}\right)}\lesssim \sqrt{\log(en)}\sum_{r:2^r\leq B_{\max,n}}2^r=O\left(B_{\max,n}\sqrt{\log(n)}\right)
\end{align*}
since, at scale $r$ there are $2^r/d_r\asymp \sqrt{2\log\left(\frac{en}{2^r}\right)}$ lengths each of which is at most $2^r$.

Therefore,
\begin{align*}
    c_E(\tau_k;B_{\max,n})\lesssim\left(\sum_{L}L\right)\sum_{\ell=1}^{\log_2(n/4)}  \frac{2^\ell}{d_\ell^2}\lesssim B_{\max,n}\sqrt{\log(n)}\sum_{\ell=1}^{\log_2(n/4)}  \frac{\log(en/2^\ell)}{2^\ell}=O\left(B_{\max,n}\log^{3/2}(n)\right)
\end{align*}
Thus,
\[c(\tau_k;B_{\max,n})=O\left(B_{\max,n}\log^{3/2}(n)\right)\]
and summing over all $K_n^{(0)}$ change-points implies
\[m_{1,n}\lesssim K_{n}^{(0)}B_{\max,n}\log^{3/2}(n).\]
A sufficient condition for a test to be asymptotically powerless is then 
\begin{align*}
    \frac{m_{1,n}}{x}\lesssim \frac{K_{n}^{(0)}B_{\max,n}\log^{3/2}(n)}{x}=o(1)
\end{align*}
i.e.\ $x=K_{n}^{(0)}B_{\max,n}\log^2(n)$ will suffice subject to the condition that $x\leq m_n=O\left(n\log^{5/2}(n)\right)$. 

We now repeat the above calculation for all continuous triplets of scale $O(B_{\max,n})$, i.e. calculate $c(\tau_k;B_{\max,n})$ for $I=\{(t_1,t_2,t_3\ \mid t_1<t_2<t_3\}$. As before, fix any $\tau_k$, then for $\tau_k\in (t_1,t_2]$ necessarily $\tau_k\in (t_1,t_2]$ or $\tau_k\in (t_2,t_3]$. At the scale $B_{\max,n}$, $\tau_k\in (t_1,t_2]$ implies that $\tau_k-t_1\leq B_{\max,n}$ and $t_2-\tau_k\leq B_{\max,n}$ and by a simple grid counting there are $O(B_{\max,n}^2)$ ways to cover $\tau_k$. We then adjoin an interval of length at most $B_{\max,n}$ to find that there are $O(B_{\max,n}^3)$ ways to cover $\tau_k$ in the left interval. By symmetry, this same bound holds if $\tau_k \in (t_2,t_3]$ so that $c(\tau_k;B_{\max,n})=O(B_{\max,n}^3)$. Hence,
\[m_{1,n}\lesssim \sum_{k=1}^{K_n^{(0)}}c(\tau_k;B_{\max,n})=O\left(K_n^{(0)}B_{\max,n}^3\right).\]
Moreover, for any family of triplets $\mathcal I$ (where each triplet is unique) cannot have its count number $c_{\mathcal I}(\tau_k;B_{\max,n})$ exceed $O(B_{\max,n}^3)$ because each triplet $I \in \mathcal I$ must also be in the family of continuous triplets. Therefore, for any family of triplets $\mathcal I$
\[\frac{m_{1,n}}{x}\lesssim \frac{K_{n^{(0)}}B_{\max,n}^3}{x}=o(1)\]
if $x=K_n\gg K_n^{(0)}B_{\max,n}^3\log^2(n)$, subject to the condition that $|\mathcal I|\geq K_n^{(0)}B_{\max,n}^3\log^2(n)$. Thus, the proof concludes. 

\section{Necessary and Sufficient Conditions for Detection of High Energy Segments}\label{subsec:EnergySegments}

As discussed by \cite{Verzelen2023} the energy of a single change-point is sensitive to small perturbations, so that they introduced the notion of a high-energy segment. In this subsection, we discuss, much in the same way as Sections \ref{sec:PowerNecandSuff} and \ref{sec:MinMaxImpossiblity} for individual change-points, the necessary and sufficient conditions for detecting high-energy segments. 

To do so, we first make the key assumption that the high-energy segments are stable under end-point perturbation. A segment, $I=[a,b]$, has energy  $\mathcal E_I$ defined as
\begin{align}\label{eq:energydefSeg}
    \mathcal E_I=\max_{i\in[a,b]}\left|\frac{\sum_{j=i}^b\mu_j}{b-i+1}-\frac{\sum_{j=a-1}^{i-1}\mu_j}{i-a+1} \right|\sqrt{\frac{(b+1-i)(i+1-a)}{b-a+2}}=:\max_{i\in[a,b]}E_I(a,i,b)
\end{align}
which is the best approximation of $\mu \cap I$ by a change-point arrangement containing only one jump.  

Our notion of stability is heuristically that if the maximum energy of a segment $[a,b]$ is achieved for $i=i^*$, i.e., $\mathcal E_I = E_I(a,i^*,b)$
then $E_I(s,m,e)\approx \mathcal E_I$ if $s \approx a$, $m \approx i^*$ and $e \approx b$. This is formalized by an interval being stable in the sense of the following definition.

\begin{definition}\label{def:DeltaStable}
   Consider the high energy segment $I=(a-1,b+1]$ with maximum energy $ \mathcal E_I = E_I(a,k^*,b)$ for some $k^*\in (a-1,b+1]$ with left and right halves $R_k^c=(a-1,k^*]$, and $R_k^c=(k^*,b+1]$ respectively, such that $|L_k^c|\wedge|R_k^c|=o(n)$.
    
    We say that $I_k$ is locally CUSUM-stable if there exists a modulus
$\omega_{k,n}:[0,1)\to[0,\infty)$ with $\omega_{k,n}(\delta)\to 0$ as
$\delta\downarrow0$ such that for any intervals $L,R$ satisfying
\[\frac{|L\Delta L_k^c|}{|L_k^c|}\vee
\frac{|R\Delta R_k^c|}{|R_k^c|}\leq \delta,\]
we have
\[\Gamma(L,R)\geq \Gamma(L_k^c,R_k^c)-\omega_{k,n}(\delta)\]
where for adjacent and disjoint intervals $L,R\subset \{1,\dots, n\}$
\[\Gamma(L,R):=\left|\bar\mu_R-\bar\mu_L\right|
\sqrt{\frac{|L||R|}{|L|+|R|}}.\]
\end{definition}

If we think of Bonferroni-triplets as an $\epsilon$-net for the set of all triplets, and the energy being some functional on the space of all triplets, then Definition \ref{def:DeltaStable} is a way to control the approximation error of the $\epsilon$-net. 

The notion of a high energy segment is useful since it allows for an accumulation of undetectable change-points to be detectable if, for example, the jumps are all of positive height. To demonstrate that the generalisation of change-point detection to segment detection is necessary, we refer the interested reader to the Appendix \ref{subsec:EnergySegments} where we formally construct a detectable segment made up of individually undetectable change-points. Heuristically, we construct a staircase of tightly spaced change-points centred around $\lfloor n/2 \rfloor$ with a few change-points outside of this neighbourhood ``very far away''. 
\begin{proposition}\label{prop:HeSegSmallCP}
    There exist arrangements of individually undetectable change-points, e.g.\ those with energy less than the necessary amount in equation \eqref{eq:smallenergy}, that constitute a detectable segment, e.g.\ one satisfying the sufficient detection conditions of Theorem \ref{thm:HESegDetect}.
\end{proposition}

The following detection theorem for high energy segments is exactly analogous to Theorem \ref{thm:power}. Not only are the theorem statements similar, but also the proof strategy is the same; we construct large independent certificate sets of segments with sufficiently high energy so that all are simultaneously rejected with high probability. In fact, we consider the slightly more challenging problem where high energy segments are not nested (hence minimal) and have maximum overlap $D_n$, e.g.\ at most $D_n$ high energy intervals have a non-empty intersection. The minimality assumption is made without loss of generality, since if the minimal high-energy segment, $J$ is detected, then all larger higher energy segments $I\supset J$ are also detected.

\begin{theorem}\label{thm:HESegDetect}
    Let $\mathcal I_n$ be the candidate set of locally CUSUM-stable high energy segments $I_k=(a_k-1,b_k+1]$, with the optimal midpoint being $i_k^*$.
     Suppose that there exists a set of $K_n^S$ intervals, $\mathcal S_n\subset \mathcal I_n$ such that for each $I\in \mathcal S_n$ there does not exist a $J\in \mathcal S_n$ such that $J\subset I$. Let $D_n=\max_{x\in [2,n]}\#\{k_i \mid x\in I_{k_i}\}$ for $I_{k_i}$ being the $k_i^{th}$ member of $\mathcal I_n$. Moreover, assume that there exists some $r_n \in (0,1)$ such that
    \begin{equation}
        \max_{1\leq k_i \leq K_n^S}\left(b_{k_i}-a_{k_i}+2\right)\leq n^{r_n}
    \end{equation}
    and $(1-r_n)\log(n)\to \infty$, where $c_n=o\left(\sqrt{\log(n)}\right)$, and satisfies equation \eqref{eq:HESegcnLOwer}. 
    Finally, assume that the moduli of regularity uniformly satisfy
    \begin{equation}\label{eq:HESegcnLOwer}
    \max_{1\leq k_i \leq K_n^S}\omega_{k_i,n}\left(\frac{C}{\sqrt{(1-r_n)\log(n)}}\right)=o(c_n).\end{equation}
    Then if
    \begin{align}\label{eq:suffsegenergy}
        \min_{k_i=1}^{K^S_n}\mathcal E_{I_{k_i}}\geq \sqrt{2\log\left(\frac{nD_n}{K_n^S}\right)}+\sqrt{2\log\left(K_n^S\right)}+c_n
    \end{align}
    with $c_n=o\left(\sqrt{\log(n)}\right)$, the LBD-FDR detects all $K_n^S$ high-energy segments with probability tending to $1$ uniformly over arrangements satisfying equation \eqref{eq:suffsegenergy}; denoted $\mathfrak T_n$. Formally,
    \begin{align}\label{eq:DetectionofHESegment}
        \lim_{n \to \infty}\inf_{\mu_n \in \mathfrak T_n}\Prob_{\mu_n}\left(\bigcap_{k_i=1}^{K_n^S}\left\{\sum_{J \in \mathcal I_c(\alpha)}\mathbbm{1}\left\{J\subset I_{k_i}\right\}\geq 1 \right\}\right)= 1
    \end{align}
    recalling that $\mathcal I_c(\alpha)$ is the set of significant intervals output by LBD-FDR.
\end{theorem}

Detection of high energy segments is a generalisation to the problem of detecting high energy change-points; this can be observed by the fact that $(\tau_{k-1},\tau_{k+1}]$ is a high energy segment if, and only if, $\tau_k$ is a high energy change-point. Therefore, the impossibility results of Subsection \ref{sec:MinMaxImpossiblity} follow through immediately with the class of all change-points arrangements being a hard subclass for the arrangements of all segments. Hence, we can obtain optimality for segment detection in the same sparse case as for change-point localisation. Specifically, this is the regime in which equation \eqref{eq:suffsegenergy} is asymptotic to $\sqrt{2\log(n)}$, and selecting $D_n,K_n^S,\delta_n$ to be subpolynomial in $n$ achieves this threshold. Hence, the LBD-FDR sharply obtains the optimal constant, in sparse arrangements, for the detection of high-energy segments. 

We now present the proofs of Proposition \ref{prop:HeSegSmallCP} and Theorem \ref{thm:HESegDetect}.
 
\subsection{Proof of Proposition \ref{prop:HeSegSmallCP}}

\begin{proof}

     First take some fixed spacing, $d=10$ will do, and let $q_n=\lceil\log(n)\rceil$. Define $\ell_n=\left\lfloor\frac{n}{2}\right\rfloor-10q_n$, and pick $K_n^{\rm st}=2q_n-1$. Define the single monotone staircase by 
     \begin{equation}\label{eq:stairconstruction}
         \mu_i=\begin{cases}
             0, \quad i\leq \ell_n,\\
             jh_n, \quad \ell_n+10j<i \leq \ell_n+10(j+1), \quad j=0,\dots 2q_n-1\\
             (2q_n-1)h_n, \quad i>\ell_n+20q_n
         \end{cases}
     \end{equation}
where $h_n$ is the jump-height to be selected so that each change-point is individually below the necessary energy in equation \eqref{eq:smallenergy}. Since we have selected $K_n\lesssim n^{0.25}$, each change-point in the construction of the stairs must satisfy
\begin{align*}
    \mathcal E_k \leq \left(1-\frac{1}{\log\log(n)}\right)\left[\sqrt{2\log\left(\frac{n}{10}\right)}-\sqrt{2\log(K_n)}\right]=u_n^{\rm cp}
\end{align*}

In particular, set $h_n=\frac{u_n^{\rm cp}}{4\times 10^{1/2}}$. Then, the interior change-points satisfy
\[\mathcal E_j=5^{1/2}h_n, \quad j=2,\dots, 2(q_n-1)\]
and the first and last satisfy $\mathcal E_1=\mathcal E_{2q_n-1}\leq 10^{1/2}h_n.$ Therefore, $\max_{1\leq k \leq 2q_n-1}\mathcal E_k\leq \frac{u_n^{\rm cp}}{4}$ so that each is below the necessary energy threshold. Consider the whole clustered segment $I_n=(\ell_n,\ell_n+20q_n]$ and take the split at the midpoint $i_n^*=\ell+10q_n$ so that 
\[L_n^c=(\ell_n,\ell_n+10q_n],\quad R_n^c=(\ell_n+10q_n,\ell_n+20q_n]\]
hence, $\bar \mu_{L_n^c}=\frac{q_n-1}{2}h_n$ and $\bar \mu_{R_n^c}\frac{3q_n-1}{2}h_n$. This choice of $i_n^*$ is in fact the maximizer, since for any $1\leq j \leq 2q_n-1$
\[\bar \mu_L=\frac{j-1}{2}h_n,\quad \bar\mu_R=\frac{j+2q_n-1}{2}h_n\]
so that,
\[E(j)=q_nh_n\sqrt{\frac{100j(2q_n-j)}{20q_n}}\]
which is maximized at $j=i_n^*=q_n.$ The energy of this segment thus satisfies
\[\mathcal E_{I_n} =E_{I_n}(\ell_n+1,i_n^*,\ell_n+20q_n)=5^{1/2}q_n^{3/2}h_n=\frac{u_n^{\rm cp}}{4\times 2^{1/2}}q_n^{3/2}.\]
Since $q_n=\lceil\log(n)\rfloor\to \infty$, and $u_n^{\rm cp}=O(\sqrt{\log(n)})$, then 
\[\mathcal{E}_{I_n}=\Omega\left(\log^2(n)\right)\gg \sqrt{2\log\left(\frac{nD_n}{K_n^S}\right)}+\sqrt{2\log(K_n^S)}+c_n,\]
so that this segment has an energy above the threshold for all $n$ sufficiently large.

We now demonstrate CUSUM stability. In $I_n$ the signal takes values in $[0,H_n]$ where $H_n=(2q_n-1)h_n$, and consider intervals $L,R$ such that 
\[\frac{|L\Delta L_n^c|}{|L_n^c|}\vee\frac{|R\Delta R_n^c|}{|R_n^c|}\leq \epsilon.\]
In particular, this implies 
\begin{align*}
    |L|\in [(1-\epsilon)|L_n^c|,(1+\epsilon)|L_n^c|],\qquad |R|\in [(1-\epsilon)|R_n^c|,(1+\epsilon)|R_n^c|]
\end{align*}
Moreover, basic algebra yields
\begin{align*}
    |\bar \mu_L-\bar \mu_{L^*_n}|=\frac{|L^*_n|}{|L|}\left|\frac{1}{L^*_n}\sum_{i\in L\Delta L^*_n}\mu_i\right|\leq \frac{H_n\epsilon}{1-\epsilon} \implies |\bar \mu_R-\bar \mu_{R^*_n}|\leq \frac{H_n\epsilon}{1-\epsilon}.
\end{align*}
In particular, this implies 
\begin{align}
    \frac{1}{|\bar \mu_{L_n^c}-\bar \mu_{R_n^c}|}\frac{2H_n\epsilon}{1-\epsilon}&= \frac{2(2q_n-1)h_n}{q_nh_n}\frac{\epsilon}{1-\epsilon}\leq 4\frac{\epsilon}{1-\epsilon}=O(\epsilon)
\end{align}
\begin{align}\label{eq:meanstability}
    |\bar \mu_L-\bar \mu_R|\geq |\bar \mu_{L_n^c}-\bar \mu_{R_n^c}|-\frac{2H_n\epsilon}{1-\epsilon}=|\bar \mu_{L_n^c}-\bar \mu_{R_n^c}|(1-O(\epsilon)).
\end{align}
Finally,
\begin{align}\label{eq:geomstability}
    \sqrt{\frac{|L||R|}{|L|+|R|}}\geq  \sqrt{\frac{(1-\epsilon)^2|L_n^c||R_n^c|}{(1+\epsilon)(|L_n^c|+|R_n^c|)}}=(1-O(\epsilon))\sqrt{\frac{|L_n^c||R_n^c|}{|L_n^c|+|R_n^c|}}.
\end{align}
Combining equations \eqref{eq:meanstability} and \eqref{eq:geomstability} implies that 
\begin{align}
    \Gamma(L,R)\geq (1-O(\epsilon))\Gamma(L_n^c,R_n^c).
\end{align}

Finally, we add in a subpolynomial number of change-points a distance $n^{1-o(1)}$ away from the staircase construction so that the segment spacing assumption is satisfied. In particular, we add $\lfloor\log(n)\rfloor$ undetectable change-points to the left and right of the staircase at an even spacing on the remaining intervals $[1,\ell_n-10q_n]$, $(\ell_n+10q_n,n]$. Thus, the proof concludes.
\end{proof}

\subsection{Proof of Theorem \ref{thm:HESegDetect}}

\begin{proof}
As we aim to emulate the proof of Theorem \ref{thm:power} we first construct large independent certificate sets by grouping minimal high-energy segments into disjoint sets using the fact that they are minimal and have depth $D_n$. First, we order the segments in $\mathcal I_n^{\min}$ so that $a_1<a_2<\dots <a_{K_n^S}$, thus by minimality $b_1<b_2<\dots<b_{K_n^S}$. Next, we claim $I_j\cap I_{j+D_n}=\emptyset$. For the sake of contradiction, suppose not. Thus $a_{j+D_n}\leq b_j$, and consider the point $x=a_{j+D_n}\geq a_k$ for every $k\in \{j,j+1,\dots j+D_n\}$. Moreover, $x\leq b_j\leq b_k$ for all $k$, so that $x\in \bigcap_{k=j}^{D_n+j}I_k$, a contradiction.

    The indices of independent sets are then for $k\in [D_n]$
    \begin{align*}
        \mathcal C^{(k)}=\{j\in [K_n^S]\mid j \equiv k(\mod D_n)\} \implies \left|\mathcal C^{(r)}\right|\in \left\{\left\lfloor \frac{K_n^S}{D_n}\right\rfloor, \left\lceil \frac{K_n^S}{D_n}\right\rceil\right\}
    \end{align*}
    so that the $p$-values, $\{p_{t_k}\mid k \in \mathcal C^{(k)}\}$ computed from data inside triplets entirely contained within each $I_k$ are independent.

  Therefore, the sufficient condition for detecting all high energy segments, which is analogous to equation \eqref{eq:pt}, is that
    \begin{equation}\label{eq:pthrsSegment}
        \max_{k\in [D_n]}\max_{i\in C^{(k)}}p_i\leq \frac{\alpha q_n}{m_n}.
    \end{equation}
    for $q_n=\left\lfloor \frac{K_n^S}{D_n}\right\rfloor$.
    
We now demonstrate that there are Bonferroni triplets such that the CUSUM of each high-energy segment is approximated by the Bonferroni triplet. To do so, let $i^*_k$ be the optimal centre of the high-energy segment that satisfies equation \eqref{eq:suffsegenergy} (analogous to a high-energy change-point $\tau_{k_i}$). Then set $a_k^c=\max(i^*_k-\delta_k,a_k)$ and $z_k^c=\min(i^*_k+\delta_k,b_k)$ where $\delta_k$ is the smallest non-negative integer such that 
    \begin{equation}\label{eq:deltaAppEnergy}
        E_{I_k}(a_k^c,i^*_k,z_k^c)\geq \sqrt{2\log\left(\frac{nD_n}{K_n^S}\right)}+\sqrt{2\log(K_n^S)}+c_n.
    \end{equation}
    This $\delta_k$ always exists because by selecting $a_k^c=a^*$ and $z_k^c=b^*$ equation \eqref{eq:deltaAppEnergy} holds by assuming equation \eqref{eq:suffsegenergy}. Without loss of generality, we may assume $|L_k^c|=i^*_k-a_k^c\leq z_k^c-i^*=|R_k^c|$. Moreover, $\delta_k \in [1, b_k-a_k+2]\subset [1,n^{r_n}].$ Finally, the conclusion of equation \eqref{eq:ApproximatingEnergywithZ} holds analogously in this case, so that for all $n$ sufficiently large (and some other $c_n$ we have relabelled)
    \begin{equation*}
        E_{I_k}(a_k^c,i^*_k,z_k^c)\geq z_{\frac{\alpha q_n}{2m_n}}+\sqrt{2\log(K_n^S)}+c_n.
    \end{equation*}

    Lemma B.1 of \cite{jang2024fastoptimalchangepointdetection} implies that there exists a Bonferroni interval $(s_k,m_k]\subset (a_k^c,i^*]$ such that 
    \[\frac{(s_k-a_k^c)+(i^*_k-m_k)}{i^*_k-a_k^c}\leq \frac{8}{\sqrt{2(1-r_n)\log(n)}}=:\epsilon_n,\]
    where $\epsilon_n \to 0$ by assumption. Thus, equations (27)-(29) of \cite{jang2024fastoptimalchangepointdetection} also hold. Specifically that 
    \begin{eqnarray}
        \frac{i^*_k-m_k}{z_k^c-m_k}\leq \epsilon_n \qquad \frac{z_k^c-e_k}{z_k^c-m_k}\le \frac{\epsilon_n}{2}\qquad e_k-i_k^*\geq (1-C\epsilon_n)(z_k^c-i_k^*)
    \end{eqnarray}
    for an absolute $C>0$. Moreover, by algebraic manipulation
    \begin{align*}
        \frac{i_k^*-m_k}{|R_k^c|+i_k^*-m_k}&=\frac{i_k^*-m_k}{z_k^c-m_k}\leq \epsilon_n\implies \frac{i_k^*-m_k}{|R_k^c|}\leq \frac{\epsilon_n}{1-\epsilon_n}\\
        \iff& i_k^*-m_k+|R_k^c|\leq \frac{|R_k^*|}{1-\epsilon_n}.
    \end{align*}
    Therefore,
    \begin{align*}
        \frac{|(m_k,e_k]\Delta (i_k^*,z_k^c]|}{|R_k^c|}\leq \frac{3\epsilon_n}{1-\epsilon_n}\leq 4\epsilon_n
    \end{align*}
    where the final inequality only holds for all $n$ sufficiently large.
   
Next define $Q( X)=\left(\bar X_{(s,m]}-\bar X_{(m,e]}\right)\sqrt{\frac{(m-s)(e-m)}{e-s}}$, and in distribution,
\begin{align*}
    Q(X)\stackrel{d}{=}N\left(Q_t(\mu),1\right)
\end{align*}
where $Q_t(\mu)=\left(\bar \mu_{(s,m]}-\bar \mu_{(m,e]}\right)\sqrt{\frac{(e-m)(m-s)}{e-s}}$. Therefore, by the local CUSUM-stability property
\begin{align}\label{eq:QxDist}
    |Q_{ t_k}(\mu)|=\Gamma(L_{t_k},R_{t_k})\geq E_{I_k}(a_k^c,i^*_k,z_k^c)-\omega_{k,n}(\delta_n).
\end{align}
Applying equation \eqref{eq:QxDist} to $T_t( X)$ in equation \eqref{eq:CUSUMGaussian} demonstrates that 
\begin{align}\nonumber
    T_{t_k}( X)&\geq E_{I_k}(a_k^c,i^*_k,z_k^c)-\omega_{k,n}(\delta_n)-|Z_{t_k}|\\\label{eq:z_statSeg}
    &\geq z_{\frac{\alpha q_n}{2m_n}}+\sqrt{2\log(K_n^S)}+\frac{c_n}{4}-|Z_{t_k}|
\end{align}

The remainder of this proof follows that of Theorem \ref{thm:power} with equation \eqref{eq:z_statTtineq} replaced by equation \eqref{eq:z_statSeg}.
\end{proof}

\section{Computational Details}\label{subsec:compdetails}

First, we will outline the aspects of the general IndBH algorithm which we specialize to the graphical interval overlap case. Full details of the general IndBH algorithm can be found in Appendix B of \cite{nguyen2025controllingfalsediscoveryrate}. Throughout this section, for ease of notation, define $R=\big|\mathcal R^{\text{BH}}_\alpha\big|$.

In the broadest strokes, the algorithm procedes as follows:
\begin{enumerate}
    \item Reduce the graph $\D$ to the subgraph induced by the hypotheses rejected by the Benjamini-Hochberg procedure;
    \item Partition the subgraph into connected components;
    \item Compute the entire profile $t\mapsto \text{IndNum}(\D_k[U_t])$ for each component by interval scheduling;
    \item Use global or componentwise checks to classify most vertices;
    \item Use left or right restricted dynamic programs for those hypotheses that are undecided.
\end{enumerate}

We now fill in the details. IndBH rejects a hypothesis, $i \in \mathcal R^{\text{IndBH}_\D}_\alpha \iff p_i\leq \alpha \beta_i^*/m_n$ for $\beta_i^*=\max\{t \ : |I_i(t)|\geq t\}$, where $|I_i(t)|$ is defined as
\[|I_i(t)|=1+\text{IndNum}(\D_{\kappa(i)}[U_{-i,t}])+\sum_{k\neq \kappa(i)}\text{IndNum}(\D_k[U_t]),\]
with  $\D_k$ being the $k^{th}$ component of $\D$, $U_{-i,t}=\{v\not \in N_i \ \mid \ p_v\leq \alpha t/m_n \text{ and } v\neq i\}$ and $U_t=\{v\in \D \ \mid \ p_v\leq \alpha t/m_n\}$, where $N_i$ is the neighbourhood of node $i$. 

Base IndBH precomputes the table $ V$, which contains the shared quantities $V_{k,t}=N_k(t)=\text{IndNum}(\D_k[U(t)])$ for $k=1,\dots, n_c$, which is the number of components and $t=1,2,\dots, R$. We compute $\text{IndNum}(\D_k[U(t)])$ directly through interval scheduling which is feasible, since $\text{IndNum}(\D_k[U(t)])$ is simply the greatest number of disjoint triplets in block $k$, such that $p_j\leq \alpha t/m_n$.

By Proposition \ref{prop:NcalphaDkt}, for the $k^{th}$ fixed connected component $N_k(t)=\text{IndNum}(\D_k[U_t])$. Computing $\text{IndNum}(\D_k[U_t])$ is explained in Algorithm \ref{alg:interval-dp}, and broadly it computes the entire trajectory $t\mapsto N_k[t]$ for all $t=1,\dots, R$ in one pass, and is incremental interval scheduling across all Benjamini-Hochberg thresholds. In more detail

\begin{enumerate}
    \item[1.] A vertex $v$ becomes eligible at time 
    \[\text{act}[v]=\left\lceil \frac{m_np_v}{\alpha}\right\rceil,\]
    and once it is activated it stays activated.
    \item[2.] Sweep the activation times $t=1,2,\dots, |\mathcal R^{\text{BH}}_\alpha|$, and at time $t$ insert the intervals $v$, such that $\text{act}[v]=t$.
    \item[3.] When inserting an interval $v$ its best chain length is 
    \[1+\max\{\text{best chain ending before $L_v$}\}\]
    where $L_v$ is the left end-point of the interval $v$.
    \item[4.] The Fenwick tree stores, for each compressed endpoint, the best chain size ending at or before that endpoint. Therefore, $\text{BITQueryPrefixMax}{\text{BIT},\ l\_id[v]}$ gives the optimum size of a disjoint eligible set that is entirely to the left of $v$. We add $1$ for choosing $v$ and then update the value in the tree.
    \item[5.] After processing all intervals, the running maximum is $\text{IndNum}(\D_k[U_t]).$
\end{enumerate}

For the other part of the decomposition, one needs to compute $1+\text{IndNum}(\D_{\kappa(v)}[U_{-v,t}])$, which is the size of the largest independent set, in the component containing $v$, which itself contains $v$. Due to the interval overlap nature of the graph $\D$, $N_v[t]$ admits the decomposition
\[N_v[t]=\mathbbm{1}\{\text{act}[v]\leq t\}+N_{\text{left},v}[t]+N_{\text{right},v}[t]\]
so that only two more interval dynamic programs in order to compute the quantities $\text{IndNum}(\D_k[U(t)])$ restricted to intervals which are only to the left or right of $v$. This is a special case of the expensive fallback explained in Section B.5 of \cite{nguyen2025controllingfalsediscoveryrate}, and is computed in Line 25 of Algorithm \ref{alg:indbh-interval-short}. The remaining inclusion-exclusion checks of the general IndBH algorithm remain unchanged. 

\begin{algorithm}[!ht]
\Proc{PrecomputeIntervals$(L[1..m], U[1..m])$}{
  $\mathrm{coords} \gets \mathrm{sort}\!\left(\mathrm{unique}\!\left(\{L[i]\}_{i=1}^m \cup \{U[i]\}_{i=1}^m\right)\right)$\;
  $K \gets |\mathrm{coords}|$\;
  \For{$i \gets 1$ \KwTo $m$}{
    $L_{\mathrm{id}}[i] \gets \mathrm{rank}(L[i]\ \mathrm{in}\ \mathrm{coords})$\tcp*[r]{$1..K$}
    $U_{\mathrm{id}}[i] \gets \mathrm{rank}(U[i]\ \mathrm{in}\ \mathrm{coords})$\tcp*[r]{$1..K$}
  }
  $\mathrm{ord\_lower} \gets$ indices $1..m$ sorted by $(L_{\mathrm{id}}[i], U_{\mathrm{id}}[i])$\;
  \Return{$(\mathrm{coords}, K, L_{\mathrm{id}}, U_{\mathrm{id}}, \mathrm{ord\_lower})$}\;
}
\caption{Precompute endpoint compression for an interval family}
\label{alg:precompute-intervals}
\end{algorithm}

\begin{algorithm}[t]
\Fn{IntervalDPOverTimeBIT$(\mathrm{ord\_r}, n_{\mathrm{use}}, l_{\mathrm{id}}, r_{\mathrm{id}}, \mathrm{act}, T, K)$}{
  $\mathrm{out}[1..T] \gets 0$\;
  \If{$T = 0$ or $n_{\mathrm{use}} = 0$}{
    \Return{$\mathrm{out}$}\;
  }
  $\mathrm{idx} \gets \mathrm{ord\_r}[1..n_{\mathrm{use}}]$\tcp*[r]{already sorted by increasing $r_{\mathrm{id}}$, tie by $l_{\mathrm{id}}$}
  remove from $\mathrm{idx}$ all $v$ with $\mathrm{act}[v] > T$\;
  \If{$\mathrm{idx}$ is empty}{
    \Return{$\mathrm{out}$}\;
  }

  $\mathrm{counts}[1..T] \gets 0$\;
  \ForEach{$v \in \mathrm{idx}$}{
    $\mathrm{counts}[\mathrm{act}[v]] \gets \mathrm{counts}[\mathrm{act}[v]] + 1$\;
  }

  $\mathrm{proc\_order} \gets \mathrm{StableCountingSort}(\mathrm{idx}, \mathrm{key}=\mathrm{act}, \mathrm{buckets}=1..T)$\;
  initialize Fenwick tree $\mathrm{BIT}$ of length $K$ storing prefix maxima (all zeros)\;
  $\mathrm{best} \gets 0$, $\mathrm{ptr} \gets 1$\;

  \For{$t \gets 1$ \KwTo $T$}{
    \For{$j \gets 1$ \KwTo $\mathrm{counts}[t]$}{
      $v \gets \mathrm{proc\_order}[\mathrm{ptr}]$, $\mathrm{ptr} \gets \mathrm{ptr} + 1$\;
      $\mathrm{val} \gets 1 + \mathrm{BITQueryPrefixMax}(\mathrm{BIT}, l_{\mathrm{id}}[v])$\;
      $\mathrm{BITUpdatePointMax}(\mathrm{BIT}, r_{\mathrm{id}}[v], \mathrm{val})$\;
      $\mathrm{best} \gets \max(\mathrm{best}, \mathrm{val})$\;
    }
    $\mathrm{out}[t] \gets \mathrm{best}$\;
  }
  \Return{$\mathrm{out}$}\;
}
\caption{Interval DP over activation time using a Fenwick tree (prefix max)}
\label{alg:interval-dp}
\end{algorithm}

\begin{algorithm}[t]
\Proc{IndBHInterval$(\alpha,\ \{p_i\}_{i=1}^m,\ \{(L_i,U_i]\}_{i=1}^m,\ \mathrm{precomp})$}{
  $S \gets \mathrm{BHRejectSet}(\alpha, p, m)$\tcp*[r]{$S \subseteq \{1,\dots,m\}$}
  \If{$|S| = 0$}{
    \Return{$\emptyset$}\;
  }

  $\mathrm{BuildBHLocalArrays}(S, p, (L,U], \mathrm{precomp})$\tcp*[r]{computes $\{p_i^{\mathrm{BH}}\}_{i=1}^{R}$, $(L^{\mathrm{BH}}, U^{\mathrm{BH}})$, endpoint ids}
  $\mathrm{act}[v] \gets \min\!\left(\max\!\left(\left\lceil \frac{m\,p^{\mathrm{BH}}[v]}{\alpha}\right\rceil, 1\right), R+1\right)$
    for all $v \in \mathcal{R}_{\alpha}^{\mathrm{BH}}$\;

  $(C_1,\dots,C_{n_c}), \mathrm{comp}(v) \gets \mathrm{ComponentsBySweep}(S, (L,U], \mathrm{precomp.ord\_lower})$\tcp*[r]{$n_c$ interval-graph components}
  $\mathrm{isClique}[k] \gets \left(\max_{v\in C_k} L^{\mathrm{BH}}[v] < \min_{v\in C_k} U^{\mathrm{BH}}[v]\right)$
    for all $k \in \{1,\dots,n_c\}$\;

  \If{all $\mathrm{isClique}[k] = \mathrm{true}$}{
    \Return{$\mathrm{BlockIndBH}(\alpha, p^{\mathrm{BH}}, \{C_k\}_{k=1}^{n_c}, m)$}\;
  }

  \For{$k \gets 1$ \KwTo $n_c$}{
    $(\mathrm{ord\_r}[k], \mathrm{r\_sorted}[k], \mathrm{ord\_r}^{\mathrm{rev}}[k], \mathrm{rrev\_sorted}[k]) \gets \mathrm{PrepOrders}(C_k)$\;
    $N_k[1,\dots,R] \gets \mathrm{IntervalDPOverTimeBIT}(\mathrm{ord\_r}[k], |C_k|, l_{\mathrm{id}}, r_{\mathrm{id}}, \mathrm{act}, R, K)$\;
  }

  $N_{\mathrm{tot}}[t] \gets \sum_{k=1}^{n_c} N_k[t]$;\quad
  $N_{\mathrm{plus}}[t] \gets \max_{k} N_k[t]$\quad
  for all $t = 1,\dots,R$\;

  $R_{\mathrm{UB}} \gets \max\{t : t \le N_{\mathrm{tot}}[t]\}$ (or $0$ if none)\;
  $\mathrm{checked} \gets \{v : p^{\mathrm{BH}}[v] > \alpha R_{\mathrm{UB}}/m\}$\;
  $\mathcal{R}_{\alpha}^{\mathrm{IndBH}} \gets \emptyset$\;

  $R_{\mathrm{LB}} \gets \max\{t : t \le N_{\mathrm{tot}}[t] - N_{\mathrm{plus}}[t] + 1\}$ (or $0$ if none)\;
  $\mathcal{R}_{\alpha}^{\mathrm{IndBH}} \gets \mathcal{R}_{\alpha}^{\mathrm{IndBH}}
    \cup \{v : p^{\mathrm{BH}}[v] \le \alpha R_{\mathrm{LB}}/m\}$\;
  $\mathrm{checked} \gets \mathrm{checked} \cup \mathcal{R}_{\alpha}^{\mathrm{IndBH}}$\;

  \For{$k \gets 1$ \KwTo $n_c$}{
    $R_{\mathrm{LB}}^{(k)} \gets \max\{t : t \le N_{\mathrm{tot}}[t] - N_k[t] + 1\}$ (or $0$ if none)\;
    $\mathcal{R}_{\alpha}^{\mathrm{IndBH}} \gets \mathcal{R}_{\alpha}^{\mathrm{IndBH}}
      \cup \{v \in C_k \setminus \mathrm{checked} : p^{\mathrm{BH}}[v] \le \alpha R_{\mathrm{LB}}^{(k)}/m\}$\;
    $\mathrm{checked} \gets \mathrm{checked} \cup (\mathcal{R}_{\alpha}^{\mathrm{IndBH}} \cap C_k)$\;
  }

  \ForEach{$v \in \{1,\dots,R\} \setminus \mathrm{checked}$}{
    $k \gets \mathrm{comp}(v)$\;
    $N_v[1,\dots,R] \gets \mathrm{RestrictedViaLRDP}(v, k, \mathrm{act}, \mathrm{ord\_r},
      \mathrm{ord\_r}^{\mathrm{rev}}, \mathrm{r\_sorted}, \mathrm{rrev\_sorted})$\;
    $R_{\mathrm{true}}(v) \gets \max\{t : t \le N_{\mathrm{tot}}[t] - N_k[t] + N_v[t]\}$ (or $0$ if none)\;
    \If{$p^{\mathrm{BH}}[v] \le \alpha R_{\mathrm{true}}(v)/m$}{
      $\mathcal{R}_{\alpha}^{\mathrm{IndBH}} \gets \mathcal{R}_{\alpha}^{\mathrm{IndBH}} \cup \{v\}$\;
    }
  }

  \Return{$\mathcal{R}_{\alpha}^{\mathrm{IndBH}}$}\;
}
\caption{IndBH on an interval overlap graph (optimized interval-graph version)}
\label{alg:indbh-interval-short}
\end{algorithm}
\clearpage

\subsection{Proof of Proposition \ref{prop:LBD-FDRcomplexity}}\label{subsec:proofmemorycomplex}

\begin{proof}
    First, we define notation. Let there be $n$ data-points, $m_n=O\left(n\log^{5/2}(n)\right)$ hypotheses, $R$ the number of hypotheses rejected by the Benjamini-Hochberg procedure, $C_j$, $j=1,2,\dots n_c$ be the components of the subgraph induced by the Benjamini-Hochberg procedure, with $s_j=|C_j|$, and let $U_d$ be the number of undecided hypotheses (those $v$ looped in Line 23 of algorithm \ref{alg:indbh-interval-short}). The proof proceeds by examining the time and memory complexity of each step.

    Constructing the Bonferroni intervals takes $O(m_n)$ actions, and memory $O(m_n)$ by looping over the Bonferroni intervals and their admissible extension lengths. This is optimal since writing the intervals takes $O(m_n)$ steps.

    Once constructed, presorting the $m_n$ intervals $(t_1,t_3]$ has time complexity $O(m_n\log(m_n))$ and $O(m_n)$ memory requirement (Algorithm \ref{alg:precompute-intervals}).

    Computing the Benjamini-Hochberg rejection set, with precomputed $p$-values, has time and memory complexity
    \[T_{\text{BH}}=O(m_n\log(m_n)), \quad M_{\text{BH}}=O(m_n)\]
    by sorting the $m_n$, $p$-values. (Line 2 of algorithm \ref{alg:indbh-interval-short})

    Building the Benjamini-Hochberg local arrays (Line 5 of algorithm \ref{alg:indbh-interval-short}) has time and memory complexity at most $O(m_n)$ by extracting the intervals, $p$-values and endpoint identifications corresponding to the $p$-values rejected by the Benjamini-Hochberg procedure.

    Computing the connected components of the Benjamini-Hochberg induced subgraph has $O(m_n)$ time complexity (by sweeping over the increasing left endpoints while tracking the running maximum right endpoint). The memory requirement is also $O(m_n)$ (Line 7 of of algorithm \ref{alg:indbh-interval-short}). 

    For each component, $C_j$, sorting each of its vertices costs $O(s_j\log(s_j))$. Thus,
    \[T_{\text{comp}}=\sum_{j=1}^{n_c}s_j\log(s_j)\leq R\log\left(R\right)\]
     since $s_j\log(s_j)\leq s_j\log\left(R\right)$ and $\sum_j s_j=R$. The memory requirement is $O\left(R\right)$. This corresponds to line 8 of algorithm \ref{alg:indbh-interval-short}.

     Lines 9-10 of algorithm \ref{alg:indbh-interval-short} activate if all connected components are cliques. In this case, Proposition 8 of \cite{nguyen2025controllingfalsediscoveryrate} applies, and the rejection set is computed by running the Benjamini-Hochberg procedure on the smallest $p$-value in each connected-component. Naively, the time complexity is $O\left(R+n_c\log(n_c)\right)$ to first compute the smallest $p$-values in each component, and then sort them as part of the Benjamini-Hochberg procedure. Thus,
     \[T_{\text{clique}}=O\left(n_cR\log\left(R\right)+n_c\log(n_c)\right), \quad M_{\text{clique}}=O\left(n_c+R\right).\]
     Recall that $n_c\leq R\leq m_n$ so that $O(n_cR\log(R))=O(R^2\log(R))=O(R^2\log(n)).$

    Algorithm \ref{alg:interval-dp} applied to $s$ pre-sorted intervals, with endpoints $L$, has time complexity $T_{\text{DP}}(s)=O\left(s\log(L)+R+L\right)$. $s\log(L)$ is the complexity from querying the Fenwick tree $s$ times, $R$ is writing the output and $O(L)$ is the time complexity from the initialisation. The memory requirement is $O\left(R+L+s\right)$ (the $O\left(R\right)$ is the size of the rejection set, $O(L)$ the size of the Fenwick tree, and $O(s)$ is the number of intervals). 

    Algorithm \ref{alg:interval-dp} is called once per component, in lines 11-13 of algorithm \ref{alg:indbh-interval-short}, so the time complexity is 
    \[T_{\text{totDP}}=\sum_{j=1}^{n_c}T_{\text{DP}}(s_k)=O\left(n_c(R+L)+R\log(L)\right),\]
    with memory requirement $O\left(n_cR\right)$ by storing $N_k\left[1,\dots, R\right]$ as an array.

    Line 14 (computing the vectors $N_{\text{tot}}$ and $N_{\text{plus}}$) has time complexity $O\left(n_cR\right)$, also with memory $O\left(R\right)$ as each are vectors of that length.

    The check loop, lines 15-24 scans $n_c$ length $R$ vectors so has time complexity $O\left(n_cR\right)$ and memory $O\left(R\right)$.

    For the undecided vertices, line 25-30, two interval dynamic programs are computed (one left and one right) and one $O\left(R\right)$ calculation for $R_{\text{true}}(v)$ (line 28) so that the worst case complexity is $O\left(R+s_{\text{comp}(v)}\log(L)\right)$. Summing over undecided vertices yields time complexity of 
    \begin{align*}
        T_{\text{und}}&=O\left(U_dR\right)+\sum_{v\in U_d}O\left(s_{\text{comp(v)}}\log(L))\right)\\
        &=O\left(U_d R\right)+O\left(\sum_{j=1}^{n_c}\#\{\text{undecided in $C_j$}\}s_j\log(L)\right)\\
        &\leq O\left(U_dR\right)+O\left(\log(L)\sum_{j=1}^{n_c}s_j^2\right)= O\left(R^2\log(L)\right)
    \end{align*}

Taking the largest terms in the time and memory complexity imply
\begin{align*}
    T_{\text{tot}}&=O\left(m_n\log(m_n)+R^2\log(n)+P_t\right)\\
    M_{\text{tot}}&=O\left(m_n+L+n_cR+P_m\right),
\end{align*}
  which concludes the calculation.
\end{proof}

\end{document}

%% file: newc.tex
\newcommand{\norm}[1]{\left\lVert#1\right\rVert}

\newcommand*{\E}{E}

\newcommand*{\D}{\mathbb D}

\newcommand*{\R}{\mathbb R}
\newcommand{\Prob}{\text{pr}}

\newcommand{\newc}{\newcommand}

\newc{\V}{\mbox{V}}

\newc{\Bern}{\mbox{Bern}}
\newc{\Po}{\mbox{Po}}
\newc{\IG}{\mbox{IG}}
\newc{\Gam}{\mbox{Gam}}
\newc{\bdp}{\mathsf{p}}
\newc{\bdt}{\mathbf{t}}
\newc{\Normal}{\mathcal{N}}
\newc{\Expectation}{\mathbb{E}}
\newc{\odds}{\mbox{odds}}
\newc{\OR}{\mbox{OR}}
\newc{\stderr}{\mbox{s.e.}}
\newc{\logit}{\mbox{logit}}
\newc{\sign}{\mbox{sign}}
\newc{\SD}{\mbox{SD}}
\newc{\bdmu}{\mbox{\boldmath $\mu$}}
\newc{\bdSigma}{\mbox{\boldmath $\Sigma$}}
\newc{\bdLambda}{\mbox{\boldmath $\Lambda$}}
\newc{\bdmuhat}{\mbox{\boldmath $\hat{\mu}$}}
\newc{\bdeta}{\mbox{\boldmath $\eta$}}
\newc{\bdtheta}{\mbox{\boldmath $\theta$}}
\newc{\bdbeta}{\mbox{\boldmath $\beta$}}
\newc{\bdgamma}{\mbox{\boldmath $\gamma$}}
\newc{\bdbetahat}{\mbox{\boldmath $\hat{\beta}$}}
\newc{\bdgammahat}{\mbox{\boldmath $\hat{\gamma}$}}
\newc{\bdthetahat}{\mbox{\boldmath $\hat{\theta}$}}
\newc{\bdvareps}{\mbox{\boldmath $\varepsilon$}}
\newc{\bdzero}{\mbox{\boldmath $0$}}
\newc{\bdone}{\mbox{\boldmath $1$}}
\newc{\bdnu}{\mbox{\boldmath $\nu$}}
\newc{\bdell}{\mbox{\boldmath $\ell$}}
\newc{\bdxi}{\mbox{\boldmath $\xi$}}
\newc{\bdomega}{\mbox{\boldmath $\omega$}}
\newc{\bdepsilon}{\mbox{\boldmath $\varepsilon$}}
\newc{\bdI}{\mathbf{I}}
\newc{\bdP}{\mbox{\boldmath $P$}}
\newc{\bdX}{\mbox{\boldmath $X$}}
\newc{\bdA}{\mbox{\boldmath $A$}}
\newc{\bdB}{\mbox{\boldmath $B$}}
\newc{\bdC}{\mbox{\boldmath $C$}}
\newc{\bdD}{\mbox{\boldmath $D$}}
\newc{\bdG}{\mbox{\boldmath $G$}}
\newc{\bdJ}{\mbox{\boldmath $J$}}
\newc{\bdK}{\mbox{\boldmath $K$}}
\newc{\bda}{\mbox{\boldmath $a$}}
\newc{\bdb}{\mbox{\boldmath $b$}}
\newc{\bdc}{\mbox{\boldmath $c$}}
\newc{\bde}{\mbox{\boldmath $e$}}
\newc{\bdu}{\mbox{\boldmath $u$}}
\newc{\bdv}{\mbox{\boldmath $v$}}
\newc{\bdx}{\mbox{\boldmath $x$}}
\newc{\bdy}{\mbox{\boldmath $y$}}
\newc{\bdz}{\mbox{\boldmath $z$}}
\newc{\bdr}{\mbox{\boldmath $r$}}
\newc{\bdQ}{\mbox{\boldmath $Q$}}
\newc{\bdR}{\mbox{\boldmath $R$}}
\newc{\bdY}{\mbox{\boldmath $Y$}}
\newc{\bdT}{\mbox{\boldmath $T$}}
\newc{\bdW}{\mbox{\boldmath $W$}}
\newc{\bdH}{\mbox{\boldmath $H$}}
\newc{\bdL}{\mbox{\boldmath $L$}}
\newc{\bdU}{\mbox{\boldmath $U$}}
\newc{\bdV}{\mbox{\boldmath $V$}}
\newc{\Multinom}{\mbox{Multinom}}
\newc{\Var}{\mbox{Var}}
\newc{\var}{\mbox{var}}
\newc{\diag}{\mbox{diag}}
\newc{\tr}{\mbox{tr}}
\newc{\phat}{\hat{p}}
\newc{\Xbar}{\bar{X}}
\newc{\xbar}{\bar{x}}
\newc{\Ybar}{\bar{Y}}
\newc{\ybar}{\bar{y}}
\newc{\dbar}{\bar{d}}
\newc{\yhat}{\hat{y}}
\newc{\bdyhat}{\mbox{\boldmath $\hat{y}$}}
\newc{\ytil}{\tilde{y}}
\newc{\ftil}{\tilde{f}}
\newc{\Ho}{\mbox{\bf H}_o}
\newc{\Ha}{\mbox{\bf H}_a}
\newc{\phatYX}{\phat_Y - \phat_X}
\newc{\SSG}{\mbox{SSG}}
\newc{\SSB}{\mbox{SSB}}
\newc{\SSE}{\mbox{SSE}}
\newc{\SST}{\mbox{SST}}
\newc{\SSR}{\mbox{SSR}}
\newc{\SSAB}{\mbox{SSAB}}
\newc{\MSG}{\mbox{MSG}}
\newc{\MSB}{\mbox{MSB}}
\newc{\MSE}{\mbox{MSE}}
\newc{\MST}{\mbox{MST}}
\newc{\MSAB}{\mbox{MSAB}}
\newc{\dfE}{\mbox{dfE}}
\newc{\dfG}{\mbox{dfG}}
\newc{\dfB}{\mbox{dfB}}
\newc{\dfT}{\mbox{dfT}}
\newc{\dfAB}{\mbox{dfAB}}
\newc{\muhat}{\hat{\mu}}
\newc{\betahat}{\hat{\beta}}
\newc{\alphahat}{\hat{\alpha}}
\newc{\etahat}{\hat{\eta}}
\newc{\phihat}{\hat{\phi}}
\newc{\sigmahat}{\hat{\sigma}}
\newc{\cl}{\centerline}
\newc{\redtitle}[1]{ {\color{red}\und{#1}:} }
\newc{\bluetitle}[1]{ {\color{blue}\und{#1}:} }
\newc{\magentatitle}[1]{ {\color{magenta}\und{#1}:} }

\newc{\trans}{^\mathsf{T}}
\newc{\xtx}{\bdX\trans\bdX}
\newc{\xxtxx}{\bdX(\xtx)^{-1}\bdX\trans}
\newc{\argmin}{\operatornamewithlimits{argmin}}
\newc{\argmax}{\operatornamewithlimits{argmax}}
\newc{\setS}{\mathcal{S}}
\newc{\toP}{\overset{p}{\to}}
\newc{\toD}{\overset{d}{\to}}
\newc{\simIID}{\overset{\text{i.i.d.}}{\sim}}
\newc{\simIND}{\overset{\text{ind.}}{\sim}}
\newc{\thetahat}{\hat{\theta}}
\newc{\bds}{\mathbf{s}}
\newc{\tny}{\small}
\newc{\Col}{\textnormal{Col}}
\newc{\mutilde}{\tilde{\mu}}
\newc{\vecst}{\textnormal{vec}}
\newc{\indep}{\perp \!\!\! \perp}
\newc{\Cov}{\textnormal{Cov}}
\newc{\gtilde}{\tilde{g}}
\newc{\rank}{\textnormal{rank}}
\newc{\Cor}{\textnormal{Cor}}
\newc{\Lcol}{\mathcal{L}_{col}}
\newc{\Lrow}{\mathcal{L}_{row}}
\newc{\Ell}{\mathcal{L}}
\newc{\Norm}{\mathcal{N}}
\newc{\xtilde}{\tilde{x}}
\newc{\Xtilde}{\tilde{X}}
\newc{\Atilde}{\tilde{A}}
\newc{\Xhat}{\hat{X}}
\newc{\Vhat}{\hat{V}}
\newc{\Vperp}{V^\perp}
\newc{\Rhat}{\hat{R}}
\newc{\Shat}{\hat{S}}
\newc{\uhat}{\hat{u}}
\newc{\Phat}{\hat{P}}
\newc{\Yhat}{\hat{Y}}
\newc{\bpage}{\eject}
\newc{\bdg}{\mathbf{g}}
\newc{\bdZ}{\mathbf{Z}}
\newc{\Rn}{\mathbb{R}^n}
\newc{\Sigmahat}{\hat{\Sigma}}
\newc{\Ysubj}{Y_{(j)}}
\newc{\Rsubj}{R_{(j)}}
\newc{\A}{\mathcal{A}}

\newc{\Gammasubj}{\Gamma_{(j)}}
\newc{\Btilde}{\tilde{\mathcal{B}}}
\newc{\Xitilde}{\tilde{\Xi}}
\newc{\xitilde}{\tilde{\xi}}
\newc{\Gammatilde}{\tilde{\Gamma}}
\newc{\gammatilde}{\tilde{\gamma}}
\newc{\Lambdatilde}{\tilde{\Lambda}}
\newc{\Ls}{\mathcal{L}}
\newc{\Lambdasubj}{\Lambda_{(j)}}
\newc{\Msubj}{M_{(j)}}
\newc{\Lsubj}{L_{(j)}}
\newc{\Ytilde}{\tilde{Y}}
\newc{\Ytildehat}{\hat{\tilde{Y}}}

%% file: thm3updatedproof.tex
\subsection{Proof of Theorem \ref{thm:ImpossLocalScale}}\label{sec:minmaxsparseproofBkg0}

For the proof of Theorem \ref{thm:ImpossLocalScale}, we construct an arrangement of perturbative change-points that that are indistinguishable from the background arrangement. We do so by finding eligible long lengths between background change-points and calculate how many blocks of perturbative changepoints we can place there. Following this, we follow a standard argument, similar to that of Theorem 2 of \cite{ERCEJCYP2011}, or Theorem 3.1 of \cite{HallJin2010}, to demonstrate that the distributions are statistically indistinguishable.

\begin{proposition}\label{prop:TVproposition}
    Let $\mu_n^{(0)}$ be an arrangement of $1\leq K_n^{(0)}\lesssim n^q$. Fix a sequence of integers $B_{\max,n}$, $b_n=B_{\max,n}+1$, and $\delta_n$ such that
    \[M_n=\sum_{k=0}^{K_n^{(0)}}\left\lfloor\frac{(\tau_{k+1}-\tau_k-2b_n)_+}{2\delta_n}\right \rfloor\to \infty\]
as $n\to \infty$. Finally, fix a sequence $K_n \ll M_n$, the number of perturbative change-points with small energy, e.g.\ satisfying \eqref{eq:smallenergy}. Then 
    \[\lim_{n\to \infty}\norm{\mathbb P_{\mu_n^{(0)}}-\bar P_{\mu_n}}_{\rm TV}=0\]
    where $\bar P_{\mu_n}$ is the uniform mixture distribution over the class of perturbative blocks of change-points, and $\mathbb P_{\mu_n^{(0)}}$ is the distribution of a $N(\mu_n,I_n)$ multivariate Gaussian. Here, we define ``small'' as
    \begin{equation*}
        \mathcal E_k^2 \leq \begin{cases}
            \left(1-\varepsilon_n^{(1)}\right)^2\left[\sqrt{2\log\left(M_n\right)}-\sqrt{2\log(K_n)}\right]^2, \quad 1\leq K_n\lesssim M_n^q, \ 0\leq q<\frac{1}{4}\\
            \left(1-\varepsilon_n^{(2)}\right)^2\left[\log\left(M_n\right)-2\log(K_n)\right], \quad M_n^{1/4}\lesssim K_n\lesssim M_n^q, \ \frac{1}{4}\leq q<\frac{1}{2}, \\
        \end{cases}
    \end{equation*}
\end{proposition}

\begin{proof}

First, consider the case where $1\leq K_n\lesssim M_n^q, \ 0\leq q<\frac{1}{4}$. How we construct the perturbative sequence of change-points is by placing blocks of length $2\delta_n$ between background change-points $\tau_k$ and $\tau_{k+1}$, where $\tau_0=0$ and $\tau_{K_n^{(0)}}=n$. We call segment $k$ the indices inside $(\tau_k,\tau_{k+1}]$, and its corresponding ``trimmed'' length $a_{k,n}=(\tau_{k+1}-\tau_k-2b_n)_+$. One should note that 
\[M_{k,n}=\left\lfloor\frac{(\tau_{k+1}-\tau_k-2b_n)_+}{2\delta_n}\right \rfloor\]
is the number of blocks of length $2\delta_n$ that can be placed inside segment $k$. When $m_{k,m}\geq 1$, define 
\[u_{k,m}=\tau_k+b_n+2\delta_n(m-1),\quad (m=1,\dots, M_{k,n})\]
and the block $I_{k,m}=(u_{k,m},u_{k,m}+2\delta_n]$ with mid-point $t_{k,m}=u_{k,m}+\delta_n$. Finally, relabel the collections $\{I_{k,m}\}_{k,m}=\{I_{k}\}_{k=1}^{M_n}$ with corresponding centres $t_k$. Thus, the distance between each centre $t_k$ and background change-point is at least $\delta_n+b_n>B_{\max,n}$.

Trivially, $M_n\leq \frac{n}{2\delta_n}$, and the non-trivial lower bound can be derived as 
\begin{align*}
    M_n\geq \frac{1}{2\delta_n}\sum_{k=0}^{K_n^{(0)}}(\tau_{k+1}-\tau_k)_+-(K_{n^{(0)}}+1)\geq \frac{1}{2\delta_n}\left(n-2b_n(K_n^{(0)}+1)\right)-(K_{n^{(0)}}+1).
\end{align*}
Brief algebra yields
\[M_n\geq \frac{n}{2\delta_n}\left(1-\frac{2b_nK_n^{(0)}}{n}-\frac{2b_n}{n}-\frac{K_n^{(0)}\delta_n}{n}-\frac{2\delta_n}{n}\right),\]
and if $K_n^{(0)}(b_n+\delta_n)=o(n)$ then $M_n=\frac{n}{2\delta_n}(1+o(1))$. Therefore, there are parameter regimes such that $M_n\to \infty$ at nearly a linear rate.

We now construct our perturbatory change-points. For the block construction given above define the change-point jump function as 
\begin{equation}\label{eq:jumpfunctionh}h_j^{(i)}=\begin{cases}
    -\gamma_n, \quad j \in I_i^L,\\
    \gamma_n, \quad j\in I_i^R\\
    0, \text{ otherwise}
\end{cases}\end{equation}
where  
\[I_i^L=\{u_i+1,\dots, u_i+\delta_n\},\quad I_i^R=\{u_i+\delta_n+1,\dots, u_i+2\delta_n\}.\]

For $S\subset \{1,\dots,M_n\}$, such that $|S|=K_n$, define the subclass
\begin{equation}\label{eq:perturbatorysubclass}\mathcal A_n(\mu_n^{(0)},M_n,K_n)=\left\{\mu_n^{(0)}+\sum_{i\in S}h^{(i)} \ \mid \ S\subset [M_n], \ \ |S|=K_n\right\}.\end{equation}
In the following we take a uniform mixture prior over $\mathcal A_n(\mu_n^{(0)},M_n,K_n)$ and demonstrate that the mixture likelihood ratio tends to $1$ in $L^1$ under $\mu_n^{(0)}$. In the following, it is with no loss of generality to consider $\mathcal A_n(0,M_n,K_n)$ since we are considering a Gaussian location model.

Let $P_0=N(0,I_n)$ and $P_i=N(h^{(i)},I_n)$ so that their likelihood ratio is 
\begin{align*}
    L_i=\frac{dP_i}{dP_0}=\exp\left\{\gamma_n\sum_{j\in I_i}\sigma_jZ_j-\frac{1}{2}\norm{h^{(i)}}_2^2 \right\}
\end{align*}
for $Z_j\stackrel{\text{iid}}{\sim} N(0,1)$, $\sigma_j=\pm 1$ the encoded sign pattern and 
\begin{align*}
    \norm{h^{(i)}}_2^2 =\sum_{j\in I_i}\left(h_j^{(i)} \right)^2=2\delta_n\gamma_n^2.
\end{align*}
By standard facts for Gaussian random variables, $\E_{P_0}[L_i]=1$, $\E_{P_0}[L_i^2]=\exp\left\{\norm{h^{(i)}}_2^2 \right\}=\exp\{2\delta_n\gamma_n^2\}$. Crucially, the family $\{L_i\}_{i=1}^{M_n}$ are independent and identically distributed under $P_0$. Moreover, by defining
\[X_i=\frac{1}{\sqrt{2\delta_n}}\sum_{j\in I_i}\sigma_j Z_j\]
the $X_i$ are independent and identically distributed $N(0,1)$ under $P_0$ for disjoint blocks. By defining $w_n=\sqrt{2\delta_n\gamma_n^2}$ each $L_i$ is the canonical likelihood ratio
\[L_i=\exp\left\{w_nX_i-\frac{1}{2}w_n^2\right\}.\]
Moreover, select 
\[w_n=(1-\varepsilon_n)\left(\sqrt{2\log(M_n)}-\sqrt{2\log(K_n)}\right)\]
where $\varepsilon_n \to 0$ but $\varepsilon_n\left(\sqrt{2\log(M_n)}-\sqrt{2\log(K_n)}\right)\to \infty$. This is equivalent to picking 
\[\gamma_n = (1-\varepsilon_n)\frac{\sqrt{\log(M_n)}-\sqrt{\log(K_n)}}{\sqrt{\delta_n}}.\]

Next, for any subset $S\subset \{1,2,\dots,M_n\}$ such that $|S|=K_n$ consider the alternative mean 
\[\mu_n^{(S)}=\sum_{i \in S}h^{(i)}\]
and define $P_S=N\left(\mu_n^{(S)},I_n\right)$, with the likelihood ratio (with respect to $P_0$) as $L_S$. Due to the independence of the blocks, 
\[L_S=\prod_{i\in S}L_i\]
and we will define 
\[\bar L_n = \frac{1}{\binom{M_n}{K_n}}\sum_{|S|=K_n}L_S,\quad \tilde L_n = \frac{1}{\binom{M_n}{K_n}}\sum_{|S|=K_n}L_S\mathbbm{1}\left\{\max_{i\in S}X_i\leq T_n\right\}\]
where $T_n=\sqrt{2\log(M_n)}$, so that $0\leq \tilde L_n\leq \bar L_n$. Note, $\bar L_n$ is the likelihood ratio integrated against the uniform (over all choices of $K_n$ blocks from $M_n$) mixture prior, analogously $\tilde L_n$ is the truncated likelihood ratio also integrated against this prior $\pi_n$. 

To demonstrate statistical indistinguishability, we will show that
\begin{align*}
   \lim_{n\to \infty} \E_{0}\left|\bar L_n-1\right|= 0.
\end{align*}
A simple application of the triangle inequality yields
\begin{align*}
    \E_{0}\left|\bar L_n-1\right|\leq \E_0[\bar L_n-\tilde L_n]+\E_0|\tilde L_n-1|=1-\underbrace{\E_0[\tilde L_n]}_{(I)}+\underbrace{\E_0|\tilde L_n-1|}_{(II)}.
\end{align*}
The proof concludes with showing that $I\to 1$ and $II\to 0$. The Fubini-Tonelli theorem implies
\begin{align*}
    \E_0[\tilde L_n]=\frac{1}{\binom{M_n}{K_n}}\sum_{|S|=K_n}\E_0\left[L_S\mathbbm{1}\left\{\max_{i\in S}X_i\leq T_n\right\}\right].
\end{align*}
Each $X_i$ is independent and identically distributed $N(0,1)$ under $P_0$ hence
\begin{align*}
    \E_0\left[L_S\mathbbm{1}\left\{\max_{i\in S}X_i\leq T_n\right\}\right]&=\prod_{i\in S}\E_0\left[L_i\mathbbm{1}\left\{X_i\leq T_n\right\}\right]=\left(\E_0\left[L_1\mathbbm{1}\left\{X_1\leq T_n\right\}\right]\right)^{K_n}\\
    \implies \E_0(\tilde L_n)&=\left(\E_0\left[L_1\mathbbm{1}\left\{X_1\leq T_n\right\}\right]\right)^{K_n}.
\end{align*}
The standard Gaussian change of measure implies 
\[\left(\E_0\left[L_1\mathbbm{1}\left\{X_1\leq T_n\right\}\right]\right)^{K_n}=\Phi(T_n-w_n)^{K_n}\]
The choice of $w_n=T_n-\sqrt{2\log(K_n)}-\omega_n$ with $\omega_n\to \infty$ (more slowly than the other terms) implies
\begin{align*}
    \left(\E_0\left[L_1\mathbbm{1}\left\{X_1\leq T_n\right\}\right]\right)^{K_n}&=\Phi(T_n-w_n)^{K_n}\\
    &=\Phi[\sqrt{2\log(K_n)}+\omega_n]^{K_n}\\
    &=\left\{1-\left(1-\Phi\left[\sqrt{2\log(K_n)}+\omega_n\right]\right)\right\}^{K_n},
\end{align*}
which converges to $1$ if $K_n\left(1-\Phi\left[\sqrt{2\log(K_n)}+\omega_n\right]\right)\to 0$. By Mill's ratio 
\begin{align*}
    K_n\left(1-\Phi\left[\sqrt{2\log(K_n)}+\omega_n\right]\right)&\sim K_n \frac{\exp\left(-[\sqrt{2\log(K_n)}+\omega_n]^2/2\right)}{(2\pi)^{1/2}(\sqrt{2\log(K_n)}+\omega_n)}\\
    &=K_n \frac{\exp\left[-\log(K_n)-\omega_n^2/2-\omega_n\sqrt{2\log(K_n)}\right]}{(2\pi)^{1/2}[\sqrt{2\log(K_n)}+\omega_n]}\\
    &=\frac{\exp\left[-\omega_n^2/2-\omega_n\sqrt{2\log(K_n)}\right]}{(2\pi)^{1/2}[\sqrt{2\log(K_n)}+\omega_n]}\to 0
\end{align*}
hence $\E_0\tilde L_n \to 1.$

We now demonstrate that $II\to 0$. We do so by applying the Cauchy-Schwarz inequality so that 
\begin{align*}
    E_0\left(|\tilde L_n-1|\right)&\leq \left[E_0\left\{\left(\tilde L_n-1\right)^2\right\}\right]^{1/2}\\
    &=\sqrt{E_0\left(\tilde L_n^2-2\tilde L_n+1\right)}\\
    &=\sqrt{E_0\left(\tilde L_n^2\right)-2(1+o(1))+1}\\
    &=\sqrt{E_0\left(\tilde L_n^2\right)-1+o(1)}.
\end{align*}
Thus, if $E_0\left(\tilde L_n^2\right)\to 1$ the proof concludes. By expanding the definition of $\tilde L_n$ we find 
\begin{align*}
    \tilde L_n^2=\frac{1}{\left(\binom{M_n}{K_n}\right)^2}\sum_{|S|=K}\sum_{|T|=K}L_SL_T\mathbbm{1}\left\{\max_{i\in S}X_i\leq T_n, \ \max_{i\in T}X_i\leq T_n\right\}.
\end{align*}
Over each summand there are $J:=|S\cap T|$ indices used twice, and $2(K_n-J)=|S\Delta T|$ indices used once. Therefore,
\begin{align*}
    \E_0\left[\tilde L_n^2\right]&=\frac{1}{\left(\binom{M_n}{K_n}\right)^2}\sum_{|S|=K}\sum_{|T|=K}\E_0\left[L_SL_T\mathbbm{1}\left\{\max_{i\in S}X_i\leq T_n, \ \max_{i\in T}X_i\leq T_n\right\}\right]\\
    &=\frac{1}{\left(\binom{M_n}{K_n}\right)^2}\sum_{|S|=K}\sum_{|T|=K}\E_0\left[\prod_{i\in S\cap T}L_i^2\mathbbm{1}\{X_i\leq T_n\}\prod_{\ell\in S\Delta T}L_\ell\mathbbm{1}\{X_\ell\leq T_n\}\right]\\
    &=\E_{\pi_n}\left[\prod_{i\in S\cap T}\E_0\left[L_i^2\mathbbm{1}\{X_i\leq T_n\}\right]\prod_{\ell\in S\Delta T}\E_0\left[L_\ell\mathbbm{1}\{X_\ell\leq T_n\}\right]\right]\\
    &=\E_{\pi_n}\left[\E_0\left[L_1^2\mathbbm{1}\{X_1\leq T_n\}\right]^{|S\cap T|}\E_0\left[L_1\mathbbm{1}\{X_1\leq T_n\}\right]^{|S\Delta T|}\right]\\
    &=\E_{\pi_n}\left\{\left(\frac{\E_0\left[L_1^2\mathbbm{1}\{X_1\leq T_n\}\right]}{\E_0\left[L_1\mathbbm{1}\{X_1\leq T_n\}\right]^2}\right)^J\right\}\left(\E_0\left[L_1\mathbbm{1}\{X_1\leq T_n\}\right]\right)^{2K_n}.
\end{align*}
First of all, 
\begin{align*}
    \E_0\left[L_1\mathbbm{1}\{X_1\leq T_n\}\right]^{2K_n}&=\Phi(T_n-w_n)^{2K_n}\\
    &=\left\{1-\left[1-\Phi\left(\sqrt{2\log(K_n)}+\omega_n\right)\right]\right\}^{2K_n}
\end{align*}
which converges to $1$ if $2K_n\left(1-\Phi\left[\sqrt{2\log(K_n)}+\omega_n\right]\right)\to 0$, which we proved above (in the case of $K_n$ instead of $2K_n$). For notational ease let $q_n=\frac{\E_0\left[L_1^2\mathbbm{1}\{X_1\leq T_n\}\right]}{\E_0\left[L_1\mathbbm{1}\{X_1\leq T_n\}\right]^2}$, we will show $\E_{\pi_n}[q_n^J]\to 1.$ By Jensen's inequality $q_n\geq 1$ so that $J\mapsto q_n^J$ is convex. Moreover, $J\sim \text{Hypergeom}(M_n,K_n,K_n)$ so that stochastic domination by a Binomial random variable yields
\begin{align*}
    1\leq\E_{\pi_n}\left(q_n^J\right)\leq \left(1-\frac{K_n}{M_n}+\frac{K_n}{M_n}q_n\right)^{K_n}\leq \exp\left\{\frac{K_n^2}{M_n}(q_n-1)\right\}
\end{align*}
hence we will show that $\frac{K_n^2}{M_n}(q_n-1)\to 0$, which is implied by $\frac{K_n^2}{M_n}q_n\to 0$ when $K_n\leq M_n^{1/4-\eta}$, so that $\frac{K_n^2}{M_n}\to 0$ (in our assumed sparse regime). Moreover, in this regime when $\omega_n=o\left(\sqrt{\log(M_n)}\right)$ we have $T_n-2\sqrt{2\log(K_n)}-2\omega_n \to \infty$. 

Before calculating the limit, we simplify $q_n$ by using $\E_0\left[L_1\mathbbm{1}\{X_1\leq T_n\}\right]^2=\Phi(T_n-w_n)^2$ and 
\begin{align*}
    \E_0\left[L_1^2\mathbbm{1}\{X_1\leq T_n\}\right]&=\E_0\left[e^{2w_nX_1-w_n^2}\mathbbm{1}\{X_1\leq T_n\}\right]\\
    &=e^{w_n^2}\E_0\left[e^{2w_nX_1-2w_n^2}\mathbbm{1}\{X_1\leq T_n\}\right]=e^{w_n^2}\Phi(T_n-2w_n)
\end{align*}
so that 
\begin{align*}
    q_n&=\frac{e^{w_n^2}\Phi(T_n-2w_n)}{\Phi(T_n-w_n)^2}\\
    &=\frac{e^{w_n^2}\Phi(T_n-2w_n)}{\Phi\left[\sqrt{2\log(K_n)}+\omega_n\right]^2}=(1+o(1))e^{w_n^2}\Phi(T_n-2w_n).
\end{align*}
Next 
\begin{align*}
    T_n-2w_n=T_n-2T_n+2\sqrt{2\log(K_n)}+2\omega_n=-T_n+2\sqrt{2\log(K_n)}+2\omega_n\to -\infty.
\end{align*}
Therefore, by Mill's ratio
\begin{align*}
    \Phi(T_n-2w_n)=\Phi(-(2w_n-T_n))\sim \frac{e^{-(2w_n-T_n)^2/2}}{(2w_n-T_n)(2\pi)^{1/2}}.
\end{align*}
Further algebra demonstrates
\begin{align*}
    w_n^2-(2w_n-T_n)^2/2=&2\log(M_n)+2\log(K_n)+\omega_n^2-...\\...+
    &2\sqrt{4\log(M_n)\log(K_n)}-2\omega_n\sqrt{2\log(M_n)}+2\omega_n\sqrt{2\log(K_n)}+...\\
    ...-&\frac{1}{2}\big(2\log(M_n)+8\log(K_n)+4\omega_n^2-4\sqrt{4\log(M_n)\log(K_n)}-...\\
    ...+&4\omega_n\sqrt{2\log(M_n)}+8\omega_n\sqrt{2\log(K_n)}\big)\\
    =&\log(M_n)-2\log(K_n)-\omega_n^2-2\omega_n\sqrt{2\log(M_n)}.
\end{align*}
Thus,
\begin{align*}
    q_n&\sim \frac{\exp\left\{w_n^2-(2w_n-T_n)^2/2\right\}}{(2w_n-T_n)(2\pi)^{1/2}}=\frac{M_n}{K_n^2}\frac{e^{-\omega_n^2-2\omega_n\sqrt{2\log(K_n)}}}{(2w_n-T_n)(2\pi)^{1/2}}\\
    \implies \frac{K_n^2}{M_n}q_n &\sim \frac{e^{-\omega_n^2-2\omega_n\sqrt{2\log(K_n)}}}{(2w_n-T_n)(2\pi)^{1/2}}\to 0.
\end{align*}
Therefore, $\E\left(q_n^J\right)\to 1$ so that the mixture distribution is statistically indistinguishable from the global null when 
\[w_n=\mathcal E_k = (1-\varepsilon_n)\left[\sqrt{2\log(M_n)}-\sqrt{2\log(K_n)}\right].\]

Now consider the case when $M_n^{1/4}\lesssim K_n \lesssim M_n^q$ for $\frac{1}{4}\leq q <\frac{1}{2}$ and once again consider the subclass $\mathcal A_n(\mu_n^{(0)},M_n,K_n)$ defined in equation \eqref{eq:perturbatorysubclass}. As argued above, we can consider $\mu_n^{(0)}=0$ in the following. Once again, let $P_0=N(0,I_n)$ and $P_i=N(h^{(i)},I_n)$, with $h^{(i)}$ defined as in equation \eqref{eq:jumpfunctionh}, so that the likelihood ratio between $P_0$ and $P_i$ is 
\begin{align*}
    L_i=\frac{dP_i}{dP_0}=\exp\left\{\gamma_n\sum_{j\in I_i}\sigma_jZ_j-\frac{1}{2}\norm{h^{(i)}}_2^2 \right\}
\end{align*}
for $Z_j\stackrel{\text{iid}}{\sim} N(0,1)$, $\sigma_j=\pm 1$ the encoded sign pattern and 
\begin{align*}
    \norm{h^{(i)}}_2^2 =\sum_{j\in I_i}\left(h_j^{(i)} \right)^2=2\delta_n\gamma_n^2.
\end{align*}
By standard facts for Gaussian random variables, $\E_{P_0}[L_i]=1$, $\E_{P_0}[L_i^2]=\exp\left\{\norm{h^{(i)}}_2^2 \right\}=\exp\{2\delta_n\gamma_n^2\}$. Crucially, the family $\{L_i\}_{i=1}^{M_n}$ are independent and identically distributed under $P_0$. 

For any subset $S\subset \{1,2,\dots,M_n\}$ such that $|S|=K_n$ consider the alternative mean 
\[\mu_n^{(S)}=\sum_{i \in S}h^{(i)}\]
and define $P_S=N\left(\mu_n^{(S)},I_n\right)$, with likelihood ratio (with respect to $P_0$) as $L_S$. Due to the independence of the blocks, 
\[L_S=\prod_{i\in S}L_i.\]
We now place a uniform prior, $\pi_n$, over the alternatives of all subsets of $[M_n]$ of size $K_n$, i.e. $\pi_n(S)=\left(\binom{M_n}{K_n} \right)^{-1}$, so that the mixture distribution is 
\[\bar P_n=\E_{\pi_n}[P_S]\]
which implies the likelihood ratio is 
\[\bar L_n=\frac{d\bar P_n}{dP_0}=\E_{\pi_n}[L_S]=\left(\binom{M_n}{K_n} \right)^{-1}\sum_{|S|=K_n}\prod_{i\in S}L_i.\]
Consider, $\E_0[\bar L_n^2]$, we will show that this quantity converges to $1$ so that the $\chi^2$ divergence, and hence total variation distance, tends to $0$. Explicitly,
\begin{align*}
    \E_0[\bar L_n^2]=\left(\binom{M_n}{K_n} \right)^{-2}\sum_{S,S'}\E_0\left[\prod_{i\in S}L_i\prod_{j\in S'}L_j \right].
\end{align*}
Using that the blocks are independent for any fixed $S,S'$ of size $K_n$ we have
\begin{align*}
    \E_0\left[\prod_{i\in S}L_i\prod_{j\in S'}L_j \right]&=\prod_{i=1}^{M_n}\E_0\left[L_i^{\mathbbm{1}(i\in S)+\mathbbm{1}(i\in S')}\right]\\
    &=\prod_{i \in  S\cap S'}\E_0[L_i^2]\prod_{i\not\in S\cap S'}\E[L_i^0]\\
    &=\left(\E_0[L_1^2] \right)^{|S\cap S'|}.
\end{align*}
thus 
\[\E_0[\bar L_n^2]=\E_{(S,S')}\left\{\left(\E_0[L_1^2] \right)^{|S\cap S'|} \right\}=\E[e^{a_nI_n}]\]
for $a_n=2\delta_n\gamma_n^2$ and $I_n=|S\cap S'|$. $I_n$ is hypergeometric with parameters, $(M_n,K_n,K_n)$ and is stochastically dominated by a $\text{Bin} (K_n,K_n/M_n)$ random variable. Hence,
\[1\leq \E[e^{a_nI_n}]\leq \left(1-\frac{K_n}{M_n}+\frac{e^{a_n}K_n}{M_n} \right)^{K_n}\leq \exp\left\{\frac{K_n^2}{M_n}(e^{a_n}-1)\right\}.\]
We therefore have total variation convergence if $\frac{K_n^2}{M_n}(e^{a_n}-1) \to 0$, which is achieved for $a_n= \log\left(\frac{M_n}{K_n^2} \right)-\epsilon_n$ for $\epsilon_n$ diverging slower than $\log(n)$. Analogously to the proof of theorem 3.2 in \cite{jang2024fastoptimalchangepointdetection} every change-point in this construction satisfies 
\[a_n=\mathcal E_k^2 = 2\delta_n \gamma_n^2,\]
so that we need $\mathcal E_k^2= \log\left(\frac{M_n}{K_n^2} \right)-\epsilon_n$ for some $\epsilon_n \to \infty$ slower than $\log(n)$. By expanding $\mathcal E_k^2$ and using that $M_n=r_n\frac{n}{\delta_n}$ for some $r_n\in\left[\frac{1}{8},\frac{1}{4}\right]$
\begin{align*}
     \log\left(\frac{M_n}{K_n^2} \right)-\epsilon_n=\log\left(M_n\right)-2\log(K_n)-\varepsilon_n'.
\end{align*}
By selecting $\gamma_n^2=\frac{1}{2\delta_n}\left\{\log\left(M_n\right)-2\log(K_n)-\varepsilon_n'\right\}$ we obtain the claimed energy threshold.

\end{proof}